\documentclass[conference,10pt]{IEEEtran}
\IEEEoverridecommandlockouts

\usepackage[utf8]{inputenc}
\usepackage[T1]{fontenc}

\usepackage{amsmath, amsfonts, amssymb, amsthm, bbm}
\usepackage{float}
\usepackage{graphicx}
\usepackage{fontawesome5}
\usepackage{tikz}
\usetikzlibrary{matrix,fit,positioning,backgrounds,arrows.meta}
\usepackage{pgfplots}
\pgfplotsset{compat=1.18} %
\usepackage{booktabs}    %
\usepackage{multirow}
\usepackage{array}       %
\usepackage{cite}
\newcolumntype{L}[1]{>{\raggedright\arraybackslash}p{#1}}
\usepackage{algorithm}
\usepackage{algpseudocode} %
\algrenewcommand\algorithmicrequire{\textbf{Input:}}
\algrenewcommand\algorithmicensure{\textbf{Output:}}
\algnewcommand\Parameters{\item[\textbf{Parameters:}]}
\usepackage{mdframed}
\usepackage{ifthen}

\newboolean{showcomments}
\setboolean{showcomments}{false}
\ifthenelse{\boolean{showcomments}}
{ \newcommand{\mynote}[3]{
   \fbox{\bfseries\sffamily\scriptsize#1}{\small$\blacktriangleright$\textsf{\emph{\color{#3}{#2}}}$\blacktriangleleft$}
   }}
{ \newcommand{\mynote}[4]{}}

\newcommand{\done}[1]{%
  \ifthenelse{\boolean{showcomments}}%
    {{\color{blue}[DONE#1]}}%
    {}%
}

\newmdenv[
  linewidth=0.8pt,
  roundcorner=6pt,
  backgroundcolor=gray!5,
  linecolor=black,
  skipabove=10pt,
  skipbelow=10pt,
  innerleftmargin=10pt,
  innerrightmargin=10pt,
  innertopmargin=8pt,
  innerbottommargin=8pt
]{protocolbox}

\usepackage[shortlabels]{enumitem}

\usepackage{xcolor}
\usepackage{pifont}
\newcommand{\cmark}{\textcolor{green!70!black}{\ding{51}}}
\newcommand{\xmark}{\textcolor{red!80!black}{\ding{55}}}

\usepackage{hyperref}

\usepackage{tcolorbox}
\tcbuselibrary{skins, breakable}

\newtcolorbox[auto counter]{functionality}[2][]{%
  enhanced,
  breakable,
  float=htbp,
  colback=white,
  colframe=black,
  boxrule=1pt,
  arc=0pt, outer arc=0pt, %
  boxsep=0pt,
  left=8pt, right=8pt, top=8pt, bottom=8pt,
  before upper={\sloppy\noindent\textbf{Functionality \thetcbcounter} \textit{#2}\par\medskip},
  #1
}

\newtcolorbox[auto counter]{simulator}[2][]{%
  enhanced,
  float=htbp,
  colback=white,
  colframe=black,
  boxrule=1pt,
  arc=0pt, outer arc=0pt, %
  boxsep=0pt,
  left=8pt, right=8pt, top=8pt, bottom=8pt,
  before upper={\sloppy\noindent\textbf{Simulator \thetcbcounter} \textit{#2}\par\medskip},
  #1
}

\newtcolorbox[auto counter]{myprotocol}[2][]{%
  enhanced,
  float=htbp,
  colback=white,
  colframe=black!80,
  arc=3mm,
  boxrule=1pt,
  fonttitle=\bfseries,
  title={Protocol \thetcbcounter: #2},
  #1
}

\definecolor{sageGreen}{rgb}{0.44, 0.5, 0.39}
\definecolor{beamerGreen}{rgb}{0.0, 0.6, 0.3}
\definecolor{lightbluegrey}{RGB}{176,196,222}

\definecolor{takeawayaccent}{HTML}{B8860B}
\definecolor{takeawaybg}{HTML}{FFF8E1}

\newenvironment{takeaway}{%
  \par\addvspace{\medskipamount}%
  \noindent
  \begin{minipage}{\linewidth}
  \begin{mdframed}[
    linecolor=takeawayaccent,
    linewidth=3pt,
    leftline=true,
    topline=false,
    bottomline=false,
    rightline=false,
    innerleftmargin=12pt,
    innerrightmargin=10pt,
    innertopmargin=8pt,
    innerbottommargin=8pt,
    backgroundcolor=takeawaybg,
    skipabove=6pt,
    skipbelow=6pt
  ]
  \color{black}%
}{%
  \end{mdframed}
  \end{minipage}
  \par\addvspace{\medskipamount}
}

\theoremstyle{plain}
\newtheorem{theorem}{Theorem}

\newtheorem{lemma}{Lemma}
\theoremstyle{definition}
\newtheorem{definition}{Definition}

\newenvironment{acks}{\section*{Acknowledgments}}{}

\newenvironment{openscience}{\section*{Open Science}}{}
\newenvironment{ai}{\section*{LLM Usage Considerations}}{}

\begin{document}

\title{\Giskard : Byzantine Robust and Private Aggregation for Large-Scale Decentralized Learning}

\author{

\IEEEauthorblockN{
    Ousmane Touat\IEEEauthorrefmark{1},
    C\'esar Sabater\IEEEauthorrefmark{1},
    Mohamed Maouche\IEEEauthorrefmark{2},
    Sonia Ben Mokhtar\IEEEauthorrefmark{1}
}
\IEEEauthorblockA{
    \IEEEauthorrefmark{1} CNRS, INSA Lyon, Université Lyon 1, LIRIS, UMR5205\\
    \IEEEauthorrefmark{2} INRIA, INSA Lyon, Université Lyon 1, CITI Lab
}
\IEEEauthorblockA{
        \{ousmane.touat, cesar.sabater, sonia.ben-mokhtar\}@liris.cnrs.fr \\
        \{mohamed.maouche\}@inria.fr
    }

}

\newcommand{\spdz}{\textsc{SPDZ}}
\newcommand{\share}[1]{\langle #1 \rangle}
\newcommand{\mac}[1]{\gamma(#1)}
\newcommand{\field}{\mathbb{F}_p}
\newcommand{\parties}{\mathcal{P}}
\newcommand{\honest}{\mathcal{H}}
\newcommand{\adversaries}{\mathcal{A}}
\newcommand{\mainprotocol}{\Pi}
\newcommand{\cluster}{\mathcal{C}}
\newcommand{\committee}{\mathcal{K}}
\newcommand{\reveal}{\textsc{Reveal}}
\newcommand{\reshare}{\textsc{Reshare}}
\newcommand{\elect}{\textsc{Elect}}
\newcommand{\random}{\textsc{Random}}
\newcommand{\E}{\mathbb{E}}
\newcommand{\norm}[1]{\left\lVert#1\right\rVert}
\newcommand{\shareof}[1]{[\![ #1 ]\!]}

\newcommand{\modelVector}{\theta}
\newcommand{\modelVectorIt}[1]{\modelVector^{(#1)}}
\newcommand{\model}[2]{\modelVectorIt{#1}_{#2}}
\newcommand{\iter}{t}
\newcommand{\mixMatrix}{W}
\newcommand{\mixMatrixIt}[1]{\mixMatrix^{(#1)}}
\newcommand{\iterCnt}{T}
\newcommand{\mixMatrixItEl}[3]{\mixMatrixIt{#1}_{#2,#3}}
\newcommand{\oneVec}{\mathbbm{1}}
\newcommand{\modelSymb}{\theta} 
\newcommand{\globalModel}{\modelSymb^\star}
\newcommand{\graph}[1]{G^{(#1)}}
\newcommand{\edges}[1]{E^{(#1)}}
\newcommand{\iid}{i.i.d. }

\newcommand{\lTwoNorm}[1]{\|#1\|_2}
\newcommand{\secondEig}[1]{\lambda_2\left(#1\right)}
\newcommand{\mixSequence}{\mixMatrix^\star}

\newcommand{\dataset}[1]{\mathcal{D}_{#1}}
\newcommand{\globalDataset}{\mathcal{D}}
\newcommand{\loss}[3]{l_{#1}(#2,#3)}
\newcommand{\gradient}[3]{\nabla \loss{#1}{#2}{#3}}
\newcommand{\learningRate}{\eta}
\newcommand{\neighbors}[2]{\mathcal{N}^{(#1)}_{#2}}
\newcommand{\neighSize}{k}
\newcommand{\vect}{\modelVector}
\newcommand{\avgVec}{\tilde{\vect}}
\newcommand{\partySet}{\mathcal{P}}
\newcommand{\samo}{SAMO}
\newcommand{\tpratlowfpr}{TPR@1\%FPR}
\newcommand{\localUpdates}[4]{L_{#1}(#2,#3,#4)}
\newcommand{\peerTime}{p}
\newcommand{\incomingModels}[2]{\Theta^{#1}_{#2}}
\newcommand{\incomingModelsVar}[1]{\Theta_{#1}}
\newcommand{\modelASynch}[1]{\modelSymb_{#1}}
\newcommand{\sample}{x}
\newcommand{\sampleDom}{\mathcal{X}} 
\newcommand{\modelDom}{\mathcal{M}}
\newcommand{\partsOf}[1]{\wp(#1)}
\newcommand{\attribDom}{\mathcal{Z}}
\newcommand{\labelDom}{\mathcal{Y}}

\newcommand{\predictor}{P}
\newcommand{\mpEntropy}[2]{M_{\text{MPE}}(#1,#2)}
\newcommand{\attack}[2]{\attackSymb(#1,#2)}
\newcommand{\modelPredDist}[2]{#1[#2]}
\newcommand{\attackSymb}{\mathcal{A}}
\newcommand{\attackThreshold}{\tilde{\tau}}
\newcommand{\attackMPESymb}{\attackSymb_{\text{MPE}}}
\newcommand{\attackMPE}[3]{\attackMPESymb(#1,#2,#3)}
\newcommand{\testDataSet}[1]{\dataset{#1}^{t}}
\newcommand{\secparam}{\lambda}

\newcommand{\revision}[1]{\textcolor{black}{#1}}

\newcommand{\crs}{\mathsf{crs}}
\newcommand{\negl}{\mathsf{negl}}

\newcommand{\COM}{\mathsf{COM}}
\newcommand{\RenewShares}{\textsc{Renew-Shares}}
\newcommand{\GenRand}{\textsc{Gen-Rand}}
\newcommand{\BuildTree}{\textsc{BuildTree}}
\newcommand{\VSSShare}{\textsc{VSS-Share}}
\newcommand{\VSSSubShare}{\textsc{VSS-SubShare}}
\newcommand{\VSSReconst}{\textsc{VSS-Reconst}}
\newcommand{\InputShare}{\textsc{InputShare}}
\newcommand{\VerifyEval}{\textbf{VerifyEval}}
\newcommand{\VShare}{\textsc{VShare}}
\newcommand{\cm}{\mathsf{cm}}              %
\newcommand{\agg}{\mathrm{agg}}            %
\newcommand{\ch}{\mathrm{ch}}              %
\newcommand{\pa}{\mathrm{pa}}              %
\newcommand{\cnt}{\mathrm{count}}          %
\newcommand{\Left}{\mathit{left}}          %
\newcommand{\Right}{\mathit{right}}        %
\newcommand{\Reconst}{\textsc{Reconst}}    %
\newcommand{\VSetup}{\textsc{VSetup}}      %
\newcommand{\Giskard}{\textsc{Giskard}}
\newcommand{\VerAgg}{\textsc{VerAgg}}
\newcommand{\Multiply}{\textsc{Multiply}}
\newcommand{\Reshare}{\textsc{Reshare}}
\newcommand{\NoiseGen}{\textsc{NoiseGen}}
\newcommand{\BroadcastDown}{\textsc{Broadcast Down}}
\newcommand{\fieldsize}{q}
\newcommand{\mpcInp}[1]{x_{#1}}
\newcommand{\mpcFSymb}{f}
\newcommand{\mpcFunc}[1]{\mpcFSymb_{#1}}
\newcommand{\mpcFuncGlobal}{\mpcFSymb}
\newcommand{\mpcInpDom}{\mathbb{I}}
\newcommand{\mpcOutDom}{\mathbb{O}}
\newcommand{\tinylock}{%
  \begin{tikzpicture}[scale=0.12, baseline=-0.5ex]
    \draw[thick] (0,0) rectangle (1,0.8);
    \draw[thick] (0.2,0.8) -- (0.2,1.2) arc(180:0:0.3) -- (0.8,0.8);
  \end{tikzpicture}%
}

\newcommand{\secProtocol}{\Pi}
\newcommand{\secSubProtocol}{\rho}
\newcommand{\rwInstantiation}[2]{{#1}^{#2}}

\maketitle

\begin{abstract}

In decentralized learning, clients train a machine learning model while keeping their data locally and share their model parameters or gradients with a set of neighbors, making it vulnerable to both privacy attacks and model poisoning attacks.
Jointly addressing privacy and Byzantine behavior in decentralized learning is challenging. Indeed, Differential Privacy (DP) applied on the client side is known to increase the variance of honest model updates, which weakens Byzantine-robust aggregation against poisoning attacks. Distributed DP can mitigate this issue by adding noise after aggregation, but implementing it without a trusted curator requires both confidential and Byzantine-resilient computation.
Recent work addresses this challenge using secure multi-party computation (MPC) to implement robust aggregation while preserving input confidentiality.
However, these solutions scale badly: they either require all-to-all communication between participants or delegate the entire computation to a small subset of parties whose computational and communication cost grows with network size.
In this paper, we present Giskard, a scalable protocol for private and Byzantine-robust decentralized aggregation. Giskard organizes $n$ parties into a tree of committees of size $O(\log n)$  and evaluates a differentially private coordinate-wise approximate median via a committee-adapted secure quantile search over the value domain, using MPC within each committee. 
We assess Giskard both theoretically by proving its security and privacy properties and empirically through extensive experiments involving up to one million participants. 
Compared to its closest competitors, Giskard reduces per-party communication complexity by close to three orders of magnitude over the closest competitor while exhibiting comparable robustness under up to $n/4$ Byzantine parties.

\end{abstract}


\section{Introduction}
Federated learning (FL) was initially promoted as a privacy-preserving alternative to centralized training, since raw data remains on users' devices~\cite{fedavg}.
However, subsequent work has shown that FL is vulnerable to a wide range of attacks, many of which exploit the privileged role of the aggregation server~\cite{mothukuri2021survey}.
This centralization also creates a governance bottleneck: the entity operating the server retains effective control over the learning process, including the choice of the training objective, model architecture, aggregation rule.
As a result, FL may decentralize data storage, but it does not decentralize decision-making. From this perspective, decentralized learning, where participants interact directly in a peer-to-peer manner~\cite{hegedHus2019gossip,lian2017} appears to be an appealing alternative as it eliminates the need for a trusted central all-mighty orchestrator while preserving the convergence properties of centralized training~\cite{alghunaim2022unified}.

Nevertheless, decentralized learning systems face several challenges that hinder their wider adoption. First, these systems are inherently vulnerable to Byzantine participants, i.e., malicious or faulty nodes that submit arbitrary, manipulated, or erroneous model updates with the intent to corrupt the global model~\cite{fang2024byzantine}. Unlike in FL, where the server can act as a gatekeeper (albeit a trusted and potentially compromised one), the absence of a central authority in decentralized settings makes model poisoning attacks harder to detect and mitigate.

An orthogonal yet equally critical challenge concerns the privacy of participants' contributions. Even in the absence of a central server, nodes must exchange model updates (gradients or parameter vectors) that can leak sensitive information about their training data, and such direct access to parties updates, and the possibility to manipulate them broadens the attack surface enable privacy attacks potentially stronger than in FL~\cite{pasquini_security_2023}. 
Examples of such attacks include membership inference attacks~\cite{melis2019exploiting,nasr2019comprehensive,Zhang2020,aerniTramerMisleading, miaChang2024FL}, where an adversary determines whether a particular record was used during training and more critically, gradient/model inversion attacks~\cite{huang2021evaluating,mrini2024privacy, gradientinversion_reconst}, whereby an adversary reconstructs raw training samples with high fidelity directly from the exchanged gradients and reconstructed ones.

These two challenges have been extensively studied in isolation, each giving rise to its own line of defenses. 
On the one hand, robust aggregation functions have been devised to deal with Byzantine contributions. 
These aggregation functions are designed to limit the influence of malicious updates by filtering outliers or computing statistics robust to a fraction of corrupted inputs, such as coordinate-wise Median, Trimmed Mean, or Krum~\cite{trimmedmeans, krum}. 
These methods offer provable robustness guarantees up to a bounded fraction of Byzantine participants. 
More recently, the challenge of handling model poisoning attacks within heterogeneous environment was also tackled by boosting current robust aggregation rule with pre-aggregation mechanism~\cite{bucketing, allouah_fixing_2023} or by reformulating the loss function~\cite{erbani2025weightedlossapproachrobust}. 
On the other hand, one of the main protections against privacy attacks is differential privacy~\cite{wei2020federated} (DP), which introduces calibrated noise to bound what an adversary can infer from participants data. However adding noise in the training process would greatly disrupt model utility. 
To alleviate that, a little noise as possible must be added, which can be done going in the Distributed DP model, where the noise can be added after aggregating participants updates. 
This model, however, requires cryptographic techniques such as Secure  Aggregation~\cite{mansouri2023sok} or homomorphic encryption~\cite{zhang2020batchcrypt} to provide input confidentiality and ensure that individual updates are never exposed in the clear. So to guarantee end-to-end privacy while preserving model utility, we must combine differential privacy with cryptographic protection.

However, it is well known that these two challenges are fundamentally in tension: Byzantine-resilient aggregation functions require inspecting and comparing individual contributions, while privacy mechanisms are designed to hide them. 
This conflict limits the applicability of naive combinations. 
Approaches that combine robust aggregators with local DP mechanisms face impossibility results that bound the simultaneous guarantees of privacy, robustness, and utility~\cite{allouah_privacy-robustness-utility_2023}. 
Alternatively, implementing robust aggregation inside cryptographic protocols is a natural direction, but incurs a heavy computational and communication cost: aggregation rules such as Trimmed Mean or Median inherently require sorting and comparison operations, which are notoriously expensive to implement in encrypted form.

Prior work has explored MPC-based robust aggregation in federated settings~\cite{choffrut_sable,Jiang2025-xh,elsa2023,lycklama2023rofl}, where a central coordinator can structure the computation. Extending these guarantees to fully decentralized systems is considerably harder. 
Existing decentralized proposals~\cite{francez,ghavamipour_privacy-preserving_2024} involve either all nodes or a randomly elected committee performing MPC to perform robust aggregation, without adding a DP mechanism, that would still open the door to to privacy attacks on the resulting deciphered aggregate.
Furthermore, these designs yield all-to-all or all-to-committee communication topologies, where the per-node communication load grows at least linearly with the number of participants. 
This constitutes a fundamental scalability bottleneck that renders them impractical for large-scale fully decentralized learning deployments, where we may not expect any of the participant to yield enough computing and/or networking power to handle such workload.

\textbf{Our work:} In this paper, we aim at bridging this gap by proposing $\Giskard$, a scalable, robust and privacy-preserving decentralized learning system. To reach this objective, we build on scalable MPC systems designed to withstand malicious adversaries~\cite{millionsMPC,SecureShuffling} and use them to implement a differentially private, Byzantine-robust aggregation protocol.

\Giskard{} distribute secure computation across small committees to achieve sub-linear per-party communication, while ensuring that only carefully calibrated noisy statistics are revealed. Our key technical insight is a new formulation of coordinate-wise median as a sequence of private secure counting queries. This formulation enables a lightweight hierarchical aggregation protocol that combines communication efficiency, differential privacy, and Byzantine resilience.

As a result, $\Giskard$ is a scalable, load-balanced, private decentralized learning system that tolerates up to $n/4$ byzantine nodes via a hierarchical aggregation tree, matching the robustness properties of coordinate-wise median.

Our contributions are summarized as follows:
\begin{itemize}
    \item We propose $\Giskard$, a novel robust and private decentralized learning protocol that scales up to millions of parties.
    \item We prove that $\Giskard$ is secure against static malicious adversaries.
    \item We analyze $\Giskard$'s communication cost and show that the worst per-party communication is polylogarithmic in $n$ while its closest competitors have worst per-party communication complexity in the order of $O(n \log^2 n)$.
    \item We empirically evaluate $\Giskard$ on four datasets and at network sizes up to $n = 10^6$. $\Giskard$ reduces worst-case per-party communication by close to three orders of magnitude over the closest baseline, while its robust aggregation function matches the utility of formally robust aggregators.
    \item The code is publicly available\footnote{\url{https://anonymous.4open.science/r/Giskard-ByzRobustPrivatelDL-922D}}.
\end{itemize}

This paper is structured as follows. We first present the collaborative learning setting and technical preliminaries in Section~\ref{sec:preliminaries}. 
We describe the system and threat model alongside with the design goals in Section~\ref{sec:models-design}.
We derive the structural requirements for MPC-friendly robust aggregation, the counting-based design of \Giskard{} and its the full protocol in Section~\ref{sec:protocol}.
We then conduct \Giskard{} security, confidentiality, robustness, and communication cost analysis in Section~\ref{sec:analysis} before empirically validate our theoretical analysis in Section~\ref{sec:exp}. 
Finally, we discuss related work in Section~\ref{sec:related} before concluding in Section~\ref{sec:conclusion}.

\section{Background}
\label{sec:preliminaries}

This section introduces the setting in which \Giskard{} operates and the components on which it is built. 
We first formalize the collaborative learning problem (\S\ref{subsec:collab-learning}) and differential privacy guarantees (\S\ref{subsec:dp-def}).
We then recall the three technical components our protocol \Giskard{} relies on, that are robust aggregation (\S\ref{subsec:robust-agg}), cryptographic building blocks (\S\ref{subsec:mpc}) and committee formation (\S\ref{subsec:committees}).

\subsection{Collaborative Learning}
\label{subsec:collab-learning}

Collaborative learning is a distributed learning paradigm where the objective is to learn a statistical/learning model from independent computing units, each holding its own local data. 
In particular, here we consider a set of parties  $\partySet = \{1, \dots, n\}$, each holding a local data set $\dataset{i}$ , whose common objective is to learn a shared model $\globalModel$ that minimizes the following loss function:

\begin{equation}
	\mathcal{L}(\globalModel) = \frac{1}{|\mathcal{D}|}  \sum \limits_{\underset{}{i = 1}}^{n} \loss{i}{\globalModel}{\dataset{i}},
\end{equation}

where $\loss{i}$ is the local loss function of party $i$ and $\globalDataset = \bigcup_{i \in \partySet} \dataset{i}$ as defined in~\cite{bellet2017fast}.

\subsection{Differential Privacy}
\label{subsec:dp-def}

To analyze formal privacy guarantees in \Giskard{}, Differential Privacy (DP) is a standard framework.
DP provides a formal guarantee that the output of a statistical function does not depend too strongly on the contribution of any individual client.
Sample-level DP protects the addition or removal of a single record, whereas user-level (or client-level) DP, which we consider in this paper, protects the addition or removal of all records associated with a single client~\cite{geyer2018differentiallyprivatefederatedlearning,brendan2018learningdp}.
Accordingly, in user level DP, two datasets $D$ and $D'$ are neighboring, denoted $D \sim D'$, if one can be obtained from the other by adding or removing the entire local dataset of a single party~\cite{brendan2018learningdp}.

\begin{definition}{$(\varepsilon,\delta)$-DP~\cite{DKMMN06}}
\label{def:dp}

A randomized mechanism $\mathcal{M}$ satisfies $(\varepsilon,\delta)$-differential privacy (DP) if, for any neighboring datasets $D \sim D'$ and any measurable set $S$,

\begin{equation}
\Pr[\mathcal{M}(D) \in S] \leq e^\varepsilon \Pr[\mathcal{M}(D') \in S] + \delta.
\end{equation}

\end{definition}

\begin{definition}{$\rho$-Zero-Concentrated DP~\cite{bun_concentrated_dp}}
\label{def:zcdp}

A randomized mechanism $\mathcal{M}$ satisfies $\rho$-zero-concentrated differential privacy (zCDP) if, for any neighboring datasets $D \sim D'$ and any $\alpha > 1$,

\begin{equation}
D_\alpha\!\left(\mathcal{M}(D) ,|, \mathcal{M}(D')\right) \leq \rho \alpha,
\end{equation}

where $D_\alpha$ denotes the Rényi divergence of order $\alpha$.
\end{definition}

A $\rho$-zCDP mechanism satisfies $(\varepsilon,\delta)$-DP for any $\delta>0$, with~\cite{bun_concentrated_dp}

\begin{equation}
\varepsilon = \rho + 2\sqrt{\rho \log(1/\delta)}.
\end{equation}

In collaborative learning, we distinguish between mechanisms in which noise is added locally to each parties' update, referred to as Local DP, and mechanisms that add noise to global aggregates of party updates, referred to as the Central DP model if the aggregator is trusted, or Distributed DP if a cohort of parties are involved.
For the latter, we must rely on the use of cryptographic primitives.

\subsection{Byzantine Robust Collaborative Learning}
\label{subsec:robust-agg}

\subsubsection{Threat to Collaborative Learning Model}

While averaging is the most common aggregation policy in collaborative learning, it is also the most vulnerable to cases where parties fail to send correct model update.
Such failures include parties that fail to respond, have their updates corrupted by faults, or adversarially submit crafted updates that bias the aggregate. 
We refer to this class of attacks as model poisoning attacks~\cite{xie_fall_2019,baruch_little_2019}.

\subsubsection{Defense against Model Poisoning Attacks}

The literature gives a wide range of defenses against model poisoning. A first class restricts client updates to a bounded domain, preventing deviations far from the mean~\cite{cclip}. A second class deterministically detects and excludes corrupted updates from aggregation~\cite{caoFLTrustByzantinerobustFederated2021}.

A third class replaces averaging with robust statistical functions that tolerate a fraction of corrupted inputs. We focus on this third class, which is supported by the theoretical framework described below.

\begin{definition}{$\left(f,\kappa \right)$-Byzantine robustness ~\cite{allouah_fixing_2023}}
\label{def:robustagg}
  Let $f=\rho n$ with $\rho < \frac{1}{2}$ and $\kappa \geq 0$, the aggregation rule $F$ is $\left(f,\kappa \right)$-Byzantine robust if there exist a value $\kappa \in \mathbb{R}^+$ so that for any $(x_1, \dots, x_n) \in \mathbb{R}^N$ and any subset $S \subseteq [n]$ of size n-f the output of F is bounded as follow :
\begin{equation}
    \|F(x_1, \dots, x_n) - \overline{x_S}\|^2 \leq \frac{\kappa}{|S|} \sum_{i \in S} \| x_i - \overline{x_S}\|^{2}
\end{equation}
where $\overline{x_S} = \frac{1}{|S|} \sum_{i \in S} x_i$.
\end{definition}

Such an aggregator's output deviates by at most a bounded multiple of the honest dispersion, regardless of the attacker knowledge and attack strategies.

Examples of $(f, \kappa)$-Byzantine robust aggregators include coordinate-wise estimators (trimmed mean and median~\cite{trimmedmeans}) and geometry-based aggregators (Krum and Multi-Krum~\cite{krum}, Geometric Median~\cite{rfa}).

\subsection{Cryptographic Building Blocks}
\label{subsec:mpc}

\subsubsection{Multi-Party Computation} \Giskard{} relies on secure multi-party computation (MPC)~\cite{yao1982protocols}: a set of parties $P = \{1,\dots,n\}$, each holding a private input $\mpcInp{i} \in \mpcInpDom$, jointly compute a public function $\mpcFuncGlobal : \mpcInpDom^n \to \mpcOutDom^n$ such that party $i$ learns only its own output $\mpcFunc{i}((\mpcInp{j})_{j=1}^n) \in \mpcOutDom$. We instantiate this with secret-sharing-based MPC in the honest-majority setting, building on Shamir's scheme and the BGW verifiable variant. We recall both schemes below.

\subsubsection{Verifiable Secret Sharing}

Let $\mathbb{F}_\fieldsize$ be a field of size at least $\fieldsize > n$, and let $t$ denote the maximum number of corrupted MPC parties.

Shamir secret sharing~\cite{shamirSecret} represents a secret $s \in \mathbb{F}_\fieldsize$ as evaluations (denoted shares) of a random polynomial of degree $t$ whose constant term is $s$, ensuring that $t+1$ share can recover the secret, while $t$ share reveal no information about $s$.
A key property is linearity: any addition and multiplication by a public constant can be performed with local operation on the shares.

To tolerate malicious dealers and parties we use  the BGW verifiable secret-sharing (VSS) scheme~\cite{bgw88,bgwSubshare}.
VSS ensures that honest parties can reject inconsistent shares or reconstruct the secret despite malformed ones. 
BGW provides protocols for multiplying shared values and for robust reconstruction, where malformed shares submitted by corrupted parties can be detected or corrected. 
These mechanisms protect our MPC computation against active deviations by up to $t<n/3$ corrupted parties.
For completeness, we provide the full description of the BGW sharing, multiplication, and reconstruction protocols in Appendix~\ref{app:bgw-primitives}
.

\begin{theorem}[\cite{bgwSubshare}]
\label{thm:bgw}
The BGW protocol is unconditionally secure against a static malicious adversary corrupting up to $t < n/3$ parties, with a straight-line black-box simulator.
\end{theorem}

\subsection{Seed Agreement and Committee Formation}
\label{subsec:committees}

To make \Giskard{} scalable, we need to be able to form committees of a small number of parties that guarantee, with an extremely high probability, that they contain a sufficient proportion of honest parties. This requires the robust generation of a random seed, followed by the use of a public function so that each party can sample the committee membership. This is a standard construction that has been extensively studied and which we therefore reuse~\cite{millionsMPC, SecureShuffling}, based on the composition of scalable Byzantine
Agreement~\cite{kvs06, kst11, santoniFastBA}, which gives the
following.

\begin{lemma}[Common Reference String Generation]
\label{thm:quorumSize}
For all $\epsilon > 0$, there exists a protocol such that for $n$
parties in a synchronous network with authenticated channels, against
a static adversary corrupting at most $(1/4 - \epsilon)n$ parties:
\begin{itemize}
\item With probability $1 - n^{-c}$ for any constant $c$, all honest
parties agree on a common string $\mathsf{seed}$ of length
$\Theta(\log n)$ with at least a $2/3 + \epsilon$ fraction of
uniformly random bits.
\item The per-party communication cost is $\mathrm{polylog}(n)$ bits.
\item The round complexity is $\mathrm{polylog}(n)$.
\end{itemize}
\end{lemma}

The proof is adapted from composing \cite{kvs06} with \cite[Lemma 9]{santoniFastBA}. The protocol is run only once during system setup. Its output $\mathsf{seed}$ deterministically defines the tree topology via the public random functions. This topology is then reused across all subsequent operations, so no further BA invocations are needed per-operation.

\section{Models and Design Goals}
\label{sec:models-design}
We first describe the system and threat models (\S\ref{subsec:system-model}--\S\ref{subsec:threatmodel}) and finally the design goals (\S\ref{subsec:design-goals}).

\subsection{System and Network Model}
\label{subsec:system-model}

We model our collaborative learning system as a set of $n$ parties, $\parties = \{P_1, \dots, P_n\}$.
The parties jointly execute a prescribed protocol $\mainprotocol$ to collaboratively learn a global model.
We assume that the set of parties $\parties$ consists of two disjoint sets: the honest parties $\honest$, which strictly follow $\mainprotocol$, and the adversarial parties $\adversaries$ (also called Byzantine parties), with $\honest \cup \adversaries = \parties$ and $\honest \cap \adversaries = \emptyset$.
We assume a synchronous communication network where computation proceeds in rounds and messages sent by honest parties are delivered within a known bounded delay.
We further assume a public-key infrastructure (PKI) that authenticates point-to-point communication channels, and the presence of a secure communication protocol such as TLS to ensure channel confidentiality, all being necessary pre-requisite for our MPC primitives.

\subsection{Threat Model}
\label{subsec:threatmodel}

We consider a static, computationally bounded adversary $\adversaries$ who corrupts a set of parties of size $f = |\adversaries|$ prior to protocol execution.
Static corruption means that the adversary choose the corrupted parties, and cannot change this choice during the protocol execution. 
The computationally bounded assumption implies that the adversary is limited to polynomial-time computation and as such cannot break the underlying cryptographic assumptions, such as the security of authentication and encryption mechanisms.
The adversarial objective is to undermine the confidentiality, privacy, and correctness of the aggregation process.

\paragraph{Attacker capabilities.}
$\Giskard$ considers a Byzantine threat model in which corrupted parties may arbitrarily deviate from the prescribed protocol $\mainprotocol$.
At the protocol level, corrupted parties may send malformed shares, colluding to reconstruct shared secrets, or attempting to abort the protocol.
At the model learning level, corrupted parties may submit adversarially crafted model updates in order to bias the global aggregated model.

\paragraph{Corruption bound.}
We require $f < n/4$. 
This bounds results from the requirement for the optimized \Multiply{} protocol used within \Giskard{} and the committee formation requirements.

\subsection{Design Goals}
\label{subsec:design-goals}
Our primary objective is to ensure secure and reliable decentralized learning even when a subset of parties acts maliciously. Specifically, \Giskard{} is designed to achieve the following three goals:

\textbf{(1) Byzantine Robustness.}
We want our aggregation rule to be provably robust, ideally under the $(f,\kappa)$-robustness (Def.~\ref{def:robustagg}), so that the training process gets robustness guarantee against model poisoning attacks.

\textbf{(2) Privacy and Security.}
\Giskard{} must provide formal differential privacy for honest updates, while ensuring perfect, information-theoretic confidentiality against static malicious adversaries under the assumed corruption bound.
Only prescribed outputs may be revealed, enabling distributed DP without a trusted aggregator.
It must also guarantee output delivery to all honest parties despite missing, malformed, or inconsistent adversarial contributions.

\textbf{(3) Scalability.}
The per-party communication overhead must grow sub-linearly in the total number of participants, ensuring that the protocol remains practical at cross-device scale, concretely for deployments involving $n \geq 10^4$.

The main challenge is scalability, since robust rules use costly non-linear operations and processing all n inputs in one MPC instance concentrates computation and communication.
The rule and its MPC realization must therefore be co-designed to satisfy Goal~(3).

\section{\Giskard{} Protocol}
\label{sec:protocol}
\subsection{Overview and Intuition}

\subsubsection{Private Secure Counting for Median Search}
\label{subsec:counting}

Our key observation is that the coordinate-wise median, a known robust aggregation rule, can be formulated it through a series of secure queries.
Given a public pivot $p$, define the counting query
\begin{equation}
    \mathbf{count}(p) \;=\; \sum_{j=1}^{n} \mathbf{1}[x_j \leq p],
    \label{eq:counting-query}
\end{equation}
which returns how many inputs lie below $p$.
Furthermore, instead of revealing the exact count, we reveal a differentially private noisy count
\begin{equation}
    \widetilde{\mathbf{count}}
    =
    \mathbf{count}(p)+Z_t
    \qquad
    Z_t\sim\mathcal{N}\left(0,\frac{N_{\mathrm{iter}}}{2\rho}\right),
    \label{eq:noisy-counting-query}
\end{equation}
where independent Gaussian noise is added at each iteration.
Since each counting query has sensitivity one, allocating a privacy budget $\rho/N_{\mathrm{iter}}$ to each query yields an overall $\rho$-CDP guarantee over the $N_{\mathrm{iter}}$ adaptive counting queries by composition~\cite{huang2021instance}.

Algorithm~\ref{alg:median-search} gives the scalar procedure; model
coordinates are processed in parallel. For the robustness and privacy analysis, the revealed noisy count decomposes as
\begin{equation}
    \widetilde{\mathbf{count}}^{(t)}
    =
    \mathbf{count}_{H}^{(t)}+\mathbf{count}_{B}^{(t)}+Z_t,
    \label{eq:count-decomposition}
\end{equation}
with $\mathbf{count}_{H}^{(t)}=|\{j\in\mathcal{H}:x_j\leq p_t\}|$, where bit validation ensures $\mathbf{count}_{B}^{(t)}\in[0,f]$ under any Byzantine strategy, while $Z_t$ denotes the Gaussian perturbation introduced for differential privacy.

\begin{algorithm}[t]
\caption{Differentially Private Median Search via Secure Counting}
\label{alg:median-search}
\small                       %
\begin{algorithmic}[1]
\setlength{\itemsep}{0pt}   %
\setlength{\parsep}{0pt}
\Require Inputs: $x_1,\ldots,x_n\in(-u,u]$; precision $q$; DP budget $\rho$
\Ensure Approximate median $\hat{M}$
\State $\mathsf{Left}\gets-u;\ \mathsf{Right}\gets u;\ N_{\mathrm{iter}}\gets\lceil\log_2(2u/q)\rceil$  %
\For{$t=1,\ldots,N_{\mathrm{iter}}$}
  \State $p_t\gets(\mathsf{Left}+\mathsf{Right})/2$
  \State Honest $j$: $b_j(p_t)\gets\mathbf{1}[x_j\le p_t]$;\; Byzantine $j$: arbitrary $b_j(p_t)$   %
  \State Verify $b_j(p_t)\in\{0,1\}$; replace invalid by $0$
  \State Securely compute $\widetilde{\mathbf{count}}^{(t)}\gets\sum_{j=1}^n b_j(p_t)+Z_t$, where $Z_t\sim\mathcal{N}(0,\,N_{\mathrm{iter}}/2\rho)$  %
  \If{$\widetilde{\mathbf{count}}^{(t)}\ge\lfloor n/2\rfloor+1$}
    \State $\mathsf{Right}\gets p_t$
  \Else
    \State $\mathsf{Left}\gets p_t$
  \EndIf
\EndFor
\State \Return $\hat{M}\gets(\mathsf{Left}+\mathsf{Right})/2$
\end{algorithmic}
\end{algorithm}

Securely evaluating the query is inexpensive because the reformulation moves all costly operations outside the secure domain.
\emph{Local comparisons are free}, as they are performed in plaintext by the participating parties, while \emph{aggregation reduces to a sum}, an efficient associative operation that can be delegated to small committees and combined hierarchically.
Finally, \emph{Byzantine influence can be provably bounded}: each submitted contribution must be proven to be binary, and any contribution that fails this check is discarded.
Consequently, each Byzantine party can alter the validated count by at most one, so the Byzantine contribution satisfies $\mathbf{count}_{B}^{(t)}\in[0,f]$.
DP introduces an additional perturbation $Z_t$, yielding
\begin{equation}
    \widetilde{\mathbf{count}}^{(t)}-\mathbf{count}_{H}^{(t)}
    =
    \mathbf{count}_{B}^{(t)}+Z_t.
\end{equation}
Therefore, every binary-search step operates on a private and Byzantine-robust estimate of the rank, and the final approximation error can be characterized in terms of a Byzantine rank error and a privacy-induced rank error.

\subsubsection{Hierarchical Committee Architecture}
To distribute the aggregation workload, \Giskard{} organizes the parties into small committees arranged as a $k$-ary tree of depth $L$.
Figure~\ref{fig:protocol-overview} illustrates this architecture.
Each level of the tree has a specific role:
\begin{itemize}
\item \textbf{Leaves (depicted as Level 0 in Figure~\ref{fig:protocol-overview})} -- all $n$ parties. Each compares its
local vector to the current pivot coordinate-wise and secret-shares the
resulting bit vector to its assigned base committee.
\item \textbf{Base committees (depicted as Level $1$ in Figure~\ref{fig:protocol-overview})} -- each receives shares from
$n/k^L$ leaves, validates share domain and well-formedness, sums them,
and forwards re-shared sums to its parent.
\item \textbf{Intermediate committees (depicted as Levels $2$ to $L-1$ in Figure~\ref{fig:protocol-overview})} -- receive
shares from $k$ child committees, convert and re-aggregate, forward upward.
\item \textbf{Root committee (depicted as Level $L$ in Figure~\ref{fig:protocol-overview})} -- performs final checks and aggregation, adds jointly generated Gaussian noise to the count, reconstructs the noisy count in the clear, selects the next pivot, and disseminates it down the tree.
\end{itemize}

This hierarchical organization distributes the computation and communication load across multiple committees rather than concentrating it on a single one.
We next presents the complete \Giskard{} protocol and shows how these levels interact during setup and each median-search iteration.

\begin{figure*}[t]
  \centering
  \begin{tikzpicture}[
      y=0.65cm,
      every node/.style={font=\small},
      leaf/.style={circle, draw=blue!70, fill=blue!10, minimum size=20pt,  minimum width=0.7cm,
                   inner sep=0pt, font=\scriptsize},
      icommittee/.style={rectangle, rounded corners=2pt, draw=orange!70, fill=orange!10,
                      minimum width=1.15cm, minimum height=0.6cm, font=\small\bfseries},
      committee/.style={rectangle, rounded corners=3pt, draw=red!70, fill=red!15,
                        minimum width=1.6cm, minimum height=0.65cm,
                        font=\normalsize\bfseries},
      shareup/.style={->, line width=1.5pt, blue!40},
      countup/.style={->, line width=1.5pt, orange!50},
      arrowdown/.style={->, line width=1.5pt, green!50!black, dashed},
      lockinline/.style={fill=white, inner sep=1pt, font=\scriptsize, text=blue!60!black},
      countinline/.style={fill=white, inner sep=1pt, font=\scriptsize, text=orange!70!black},
      databox/.style={rectangle, rounded corners=1pt, draw=blue!50, fill=blue!5,
                      minimum width=0.7cm, minimum height=0.35cm, font=\scriptsize},
      dotsstyle/.style={font=\large, text=gray},
      phasegroup/.style={font=\small, align=right, anchor=east, text width=2.6cm},
      levellbl/.style={font=\small\itshape, text=gray},
      legendtxt/.style={font=\small, anchor=west, align=left},
      leaf/.append style={font=\scriptsize, minimum width=0.8cm, inner sep=0pt},
  databox/.append style={font=\scriptsize, minimum width=0.7cm, inner sep=0pt},
  ]

    \node[committee] (top) at (2.5, 1.25) {$\mathsf{COM}$};
    
    \foreach \i/\lbl/\x in {1/1/-2.35, 2/2/-1.65, 3/{n/b}/-0.65}{
    \node[leaf] (la\i) at (\x, -4.6) {$P_{\lbl}$};
    \node[databox] at (\x, -5.35) {$\mathbf{x}_{\lbl}$};
    }
    \node[dotsstyle] at (-1.15, -4.6) {$\cdots$};
    \node[dotsstyle] at (-1.15, -5.35) {$\cdots$};
    
    \foreach \i/\lbl/\x in {1/{i}/1.65, 2/{i{+}1}/2.35, 3/{i{+}n/b}/3.35}{
    \node[leaf] (lm\i) at (\x, -4.6) {$P_{\lbl}$};
    \node[databox] at (\x, -5.35) {$\mathbf{x}_{\lbl}$};
    }
    \node[dotsstyle] at (2.85, -4.6) {$\cdots$};
    \node[dotsstyle] at (2.85, -5.35) {$\cdots$};
    
    \node[dotsstyle, font=\LARGE] at (0.5, -4.6) {$\cdots$};
    \node[dotsstyle, font=\LARGE] at (0.5, -3.0) {$\cdots$};
    \node[dotsstyle, font=\LARGE] at (4.5, -4.6) {$\cdots$};
    \node[dotsstyle, font=\LARGE] at (4.5, -3.0) {$\cdots$};

    \foreach \i/\lbl/\x in {1/{n{-}n/b}/5.65, 2/{n{-}1}/6.65, 3/{n}/7.35}{
    \node[leaf] (lc\i) at (\x, -4.6) {$P_{\lbl}$};
    \node[databox] at (\x, -5.35) {$\mathbf{x}_{\lbl}$};
    }
    \node[dotsstyle] at (6.15, -4.6) {$\cdots$};
    \node[dotsstyle] at (6.15, -5.35) {$\cdots$};

    \node[icommittee] (c1) at (-1.5, -3.0) {$\mathcal{C}_{1}$};
    \node[icommittee] (cmid) at (2.5, -3.0) {$\mathcal{C}_{a}$};
    \node[icommittee] (cb) at (6.5, -3.0) {$\mathcal{C}_{b}$};
    
    \foreach \i in {1,2,3}{
    \draw[shareup] (la\i) -- node[lockinline, pos=0.55] {\faLock} (c1);
    }
    \foreach \i in {1,2,3}{
    \draw[shareup] (lm\i) -- node[lockinline, pos=0.55] {\faLock} (cmid);
    }
    \foreach \i in {1,2,3}{
    \draw[shareup] (lc\i) -- node[lockinline, pos=0.55] {\faLock} (cb);
    }
    
    \node[icommittee, minimum width=1.25cm] (m1) at (-0.5, -1.2) {$\mathcal{C}'_{1}$};
    \node[icommittee, minimum width=1.25cm] (mmid) at (2.5, -1.2) {$\mathcal{C}'_{a}$};
    \node[icommittee, minimum width=1.25cm] (mk) at ( 5.5, -1.2) {$\mathcal{C}'_{b}$};
    
    \node[dotsstyle, font=\LARGE] at (1.0, -1.2) {$\cdots$};
    \node[dotsstyle, font=\LARGE] at (4.0, -1.2) {$\cdots$};
    
    \draw[countup] (c1.north) -- node[countinline, pos=0.6] {\faLock\,$\scriptstyle\Sigma$} ([xshift=-14pt]m1.south);
    \draw[countup] (-0.5, -2.25) -- node[countinline, pos=0.5] {\faLock\,$\scriptstyle\Sigma$} (m1.south);
    \draw[countup] ( 0.3, -2.25) -- node[countinline, pos=0.5] {\faLock\,$\scriptstyle\Sigma$} ([xshift=14pt]m1.south);
    \node[dotsstyle, font=\small] at (-0.1, -2.45) {$\cdots$};
    
    \draw[countup] (1.7, -2.25) -- node[countinline, pos=0.5] {\faLock\,$\scriptstyle\Sigma$} ([xshift=-14pt]mmid.south);
    \draw[countup] (cmid.north) -- node[countinline, pos=0.6] {\faLock\,$\scriptstyle\Sigma$} (mmid.south);
    \draw[countup] (3.3, -2.25) -- node[countinline, pos=0.5] {\faLock\,$\scriptstyle\Sigma$} ([xshift=14pt]mmid.south);
    
    \draw[countup] (cb.north) -- node[countinline, pos=0.6] {\faLock\,$\scriptstyle\Sigma$} ([xshift=14pt]mk.south);
    \draw[countup] (5.5, -2.25) -- node[countinline, pos=0.5] {\faLock\,$\scriptstyle\Sigma$} (mk.south);
    \draw[countup] (4.7, -2.25) -- node[countinline, pos=0.5] {\faLock\,$\scriptstyle\Sigma$} ([xshift=-14pt]mk.south);
    \node[dotsstyle, font=\small] at (5.1, -2.45) {$\cdots$};
    
    \node[committee] (top) at (2.5, 1.25) {$\mathsf{COM}$};
    
    \draw[countup] (m1.north) -- (-0.5, -0.53);
    \node[dotsstyle, font=\Large] at (-0.5, -0.10) {$\vdots$};
    \draw[countup] (-0.5, 0.34)
      -- node[countinline, pos=0.55] {\faLock\,$\scriptstyle\Sigma$}
      ([xshift=-16pt]top.south);
    
    \draw[countup] (mmid.north) -- (2.5, -0.53);
    \node[dotsstyle, font=\Large] at (2.5, -0.10) {$\vdots$};
    \draw[countup] (2.5, 0.34)
      -- node[countinline, pos=0.55] {\faLock\,$\scriptstyle\Sigma$}
      (top.south);
    
    \draw[countup] (mk.north) -- (5.5, -0.53);
    \node[dotsstyle, font=\Large] at (5.5, -0.10) {$\vdots$};
    \draw[countup] (5.5, 0.34)
      -- node[countinline, pos=0.55] {\faLock\,$\scriptstyle\Sigma$}
      ([xshift=16pt]top.south);
    
    \node[dotsstyle, font=\small] at (1.2, 0.42) {$\cdots$};
    \node[dotsstyle, font=\small] at (3.8, 0.42) {$\cdots$};
    
    \draw[arrowdown] (top.west)
      to[out=-120, in=90] (-0.5, 0.34);
    \draw[arrowdown] (top.east)
      to[out=-60, in=90] (5.5, 0.34);
    \draw[arrowdown] ([xshift=5pt]top.south)
      -- ([xshift=5pt]2.5, 0.34);
    
    \draw[arrowdown] (-1.0, -0.37)
      to[out=-90, in=120] ([xshift=4pt]m1.north west);
    \draw[arrowdown] ([xshift=5pt]2.5, -0.37)
      -- ([xshift=5pt]mmid.north);
    \draw[arrowdown] (6.0, -0.37)
      to[out=-90, in=60] ([xshift=-4pt]mk.north east);
    
    \draw[arrowdown] (m1.west)
      to[out=180, in=90] (-2.9, -2.05)
      to[out=-90, in=120] (c1.north west);
    \draw[arrowdown] (mk.east)
      to[out=0, in=90] (7.9, -2.05)
      to[out=-90, in=60] (cb.north east);
    \draw[arrowdown] ([xshift=5pt]mmid.south) -- ([xshift=5pt]cmid.north);
    
    \draw[dashed, gray!70] (-3.3, -2.1) -- (8.3, -2.1);
    \draw[dashed, gray!70] (-3.3,  0.50) -- (8.3,  0.50);
    
    \node[levellbl] at (8.7, -4.6) {Level $0$};
    \node[levellbl] at (8.7, -3.0) {Level $1$};
    \node[levellbl] at (8.7, -1.2) {Level $2$};
    \node[levellbl] at (8.7, -0.10) {$\vdots$};
    \node[levellbl] at (8.7, 1.25) {Level $L$};
    
    \node[phasegroup, text=blue!60!black] at (-3.6, -4.0)
      {\textbf{Phase 1}\\
       clients hold $\mathbf{x}_j$, compare $x_{j,d}\!\leq\!p_{t,d}$, share $[\mathbf{b}_j]$};
    
    \node[phasegroup, text=orange!60!black] at (-3.6, -1.2)
      {\textbf{Phase 2}\\
       verify, aggregate $\sum\mathbf{b}_j$, re-share upward};
    
    \node[phasegroup, text=red!60!black] at (-3.6, 1.10)
      {\textbf{Phases 3--4}\\
       compute distributed noise, reveal global $\widetilde{\mathsf{count}}$, broadcast $p_{t+1}$};
    
    \node[legendtxt, text=blue!50!black] at (8.7, -3.85)
      {\faLock\,$[\mathbf{b}_j]$: shared\\bit vector};
    \node[legendtxt, text=orange!60!black] at (8.7, -2.1)
      {\faLock\,$\Sigma$: shared\\partial count};
    \draw[arrowdown] (8.7, -0.4) -- (9.2, -0.4);
    \node[legendtxt, text=green!50!black] at (9.3, -0.4)
      {broadcast\\$p_{t+1}$};

  \end{tikzpicture}
  \caption{\Giskard{} secure median computation. The $n$ parties $P_1,\dots,P_n$ are partitioned across $b = n / k^L$ base committees $\mathcal{C}_1,\dots,\mathcal{C}_b$, each of size $m$. Each base committee receives $n/b$ bit-vector shares from its assigned parties and aggregates them into a partial count. Partial counts are re-shared upward through $L$ levels of $k$-ary aggregation to the root committee $\mathsf{COM}$, which broadcasts the next threshold $p_{t+1}$ downward (dashed).}
  \label{fig:protocol-overview}
\end{figure*}

\subsection{Protocol Description}
\label{sec:giskard-protocol}

We now present the main steps of the \Giskard{} protocol, presented in Algorithm~\ref{alg:giskard-bgw}.
\Giskard{} consists of a one-time setup phase followed by a secure median-computation phase that is executed in each decentralized learning round.
Each party \(P_j\) holds a private model update \(\mathbf{x}_j \in (-u,u]^D\), where \(D\) is the model dimension and \(u\) bounds the value domain.
The parties jointly compute an approximation \(\widehat{\mathbf{M}}\) of the coordinate-wise median with target precision \(q\).
All vector operations described below are performed coordinate-wise and in parallel.

\subsubsection{One-Time Setup}
\label{sec:giskard-setup}

Before training, parties run \textsc{BuildTree} once to construct a $k$-ary tree $T$ of depth $L$, with parties as leaves and size-$m$ committees as internal nodes.
Each party learns its leaf position, base committee, and committee memberships.
Committees use threshold $\tau < m/4$ to tolerate corrupted members.
The tree is reused across all median computations, amortizing seed agreement and committee-formation costs over learning rounds.

\subsubsection{Secure Median Computation}
\label{sec:giskard-median-computation}

\begin{algorithm}[ht]
\caption{\Giskard{}}
\label{alg:giskard-bgw}
\small
\begin{algorithmic}[1]
\Require Tree $\mathcal{T}$ ($k$-ary, depth $L$); committee size $m = O(\log n)$; threshold $\tau = \lfloor (1/4 - \epsilon) m \rfloor$; precision $q$; value-domain bound $u$
\Require Each leaf $j \in [n]$ holds private $\mathbf{x}_j \in \mathbb{R}^D$
\Ensure Coordinate-wise median $\hat{M} \in \mathbb{R}^D$

\Statex \textbf{--- Setup (once) ---}
\State Run \BuildTree; instantiate $\mathcal{T}$

\Statex \textbf{--- Median computation (each FL round) ---}
\State $\mathsf{Left} \gets -u \cdot \mathbf{1}; \quad \mathsf{Right} \gets u \cdot \mathbf{1}$
\State $\mathbf{p}_1 \gets (\mathsf{Left} + \mathsf{Right}) / 2$
\For{$t = 1, \ldots, N_{\text{iter}}$}
  \Statex \hspace{.5em}\textit{--- Phase 1: input sharing ---}
  \State Each leaf $j$ computes
    $\mathbf{b}^{(t)}_j\gets
    \bigl(\mathbf{1}[x_{j,d}\leq p_{t,d}]\bigr)_{d\in[D]}$
    \State Each leaf $j$ runs
    $\InputShare(\mathbf{b}^{(t)}_j,\mathcal{C}_j)$
    \Comment{Prot.~\ref{alg:inputshare-bgw}}
  \State Each $P_i \in \mathcal{C}_j$: verify and aggregate $\to \boldsymbol{\sigma}_i$

  \Statex \hspace{.5em}\textit{--- Phase 2: hierarchical aggregation ---}
  \For{$\ell = 1$ \textbf{to} $L - 1$}
    \For{each $\mathcal{C}_{\mathrm{ch}} \to \mathcal{C}_{\mathrm{pa}}$ \textbf{in parallel}}
      \State Run $\textsc{Reshare}(\mathcal{C}_{\mathrm{ch}}, \mathcal{C}_{\mathrm{pa}})$
        \Comment{Prot.~\ref{alg:reshare-bgw}}
    \EndFor
    \State Cross-child aggregation at each $\mathcal{C}_{\mathrm{pa}}$:
    \Statex \hspace{2em} $\boldsymbol{\sigma}_r \gets \sum_{\mathrm{ch}} \boldsymbol{\sigma}'_{r, \mathrm{ch}}$ for each $P'_r \in \mathcal{C}_{\mathrm{pa}}$
  \EndFor

  \Statex \hspace{.5em}\textit{--- Phase 3: private reconstruction ---}
  \State $\boldsymbol{z}_r^{(t)} \gets \NoiseGen\!\left(\mathcal{N}\!\left(0,\frac{N_{\mathrm{iter}}}{2\rho}\right)\right)$ \Comment{Prot.~\ref{alg:noisegen-dg}}
  \State $\boldsymbol{\sigma}_r \gets \boldsymbol{\sigma}_r + \boldsymbol{z}_r^{(t)}$
  \State $\widetilde{\mathbf{count}}^{(t)} \gets \VSSReconst(\{\boldsymbol{\sigma}_r\}_{r \in \mathcal{C}_L})$ \Comment{Prot.~\ref{alg:vss-reconst-bgw}}

  \Statex \hspace{.5em}\textit{--- Phase 4: broadcast and update ---}
  \State $\mathcal{C}_L$ updates $(\mathsf{Left}, \mathsf{Right})$ coordinate-wise from $\widetilde{\mathbf{count}}^{(t)}$ and $\mathbf{p}_t$
  \State $\mathcal{C}_L$ broadcasts $\mathbf{p}_{t+1} \gets (\mathsf{Left} + \mathsf{Right})/2$ down the tree
\EndFor

\State $\hat{M} \gets (\mathsf{Left} + \mathsf{Right})/2$; broadcast down the tree
\State \Return $\hat{M}$
\end{algorithmic}
\end{algorithm}

Algorithm~\ref{alg:giskard-bgw} describe the counting-based median search of
Subsection~\ref{subsec:counting} over the \Giskard{} committee tree, with all model coordinates processed in parallel.
The binary-search rule and stopping criterion are therefore inherited directly
from Algorithm~\ref{alg:median-search}; we describe only their secure distributed
realization below.

\paragraph{Input sharing.}
At iteration $t$, each party $P_j$ computes the comparison vector
\begin{equation}
\mathbf{b}^{(t)}_j
=
\left(
\mathbf{1}[x_{j,d}\leq p_{t,d}]
\right)_{d\in[D]}
\end{equation}
and secret-shares it with its base committee $\mathcal{C}_j$ using
$\InputShare$ (Protocol~\ref{alg:inputshare-bgw}).
Committee members verify that every shared coordinate is a bit, replace invalid
contributions by a canonical sharing of zero, and locally sum the validated
shares into a secret-shared partial count.

\paragraph{Hierarchical aggregation.}
The partial counts are propagated toward the root committee using
$\textsc{Reshare}$ (Protocol~\ref{alg:reshare-bgw}) between each child--parent
committee pair.
At every level, each parent locally sums the re-shared aggregates received from
its children.
After $L-1$ levels, the root committee holds a secret sharing of the global
coordinate-wise count.

\paragraph{Count reconstruction.}
The root committee invokes $\NoiseGen$ and then $\VSSReconst$ to reconstruct noisy count
\begin{equation}
\widetilde{\mathbf{count}}^{(t)}
=
\sum_{j=1}^{n}\mathbf{b}^{(t)}_j + Z^{(t)}
\end{equation}
Thus, $\widetilde{\mathbf{count}}^{(t)}$ is exactly the vectorized form of the secure counting
query defined in Eq.~\eqref{eq:counting-query}.

\paragraph{Broadcast and update.}
The root committee applies the binary-search update of
Algorithm~\ref{alg:median-search} independently to every coordinate using
$\widetilde{\mathbf{count}}^{(t)}$ and $\mathbf{p}_t$.
It then computes $\mathbf{p}_{t+1}$ and broadcasts the new public pivot down the committee tree.
After $N_{\mathrm{iter}}$ iterations, the final midpoint is returned as the aggregated model $\hat{M}$.
The protocol then reuses the same committee topology for the next round.
Algorithm~\ref{alg:giskard-bgw} specifies how the protocol components are composed, but treats their internal executions as black boxes.
We next describe \textsc{BuildTree}, $\InputShare{}$, $\NoiseGen$ and $\Reshare$.

\subsection{Sub-Protocols}
We now detail the four sub-protocols invoked by \Giskard{} (Algorithm~\ref{alg:giskard-bgw}): $\BuildTree$ (line~1), $\InputShare{}$ (line~6), $\Reshare$ (line~10) and $\NoiseGen$ (line~14).

\subsubsection{$\BuildTree$ for Topology Generation}

\BuildTree{} is a one-time protocol that randomly constructs the communication topology through seed agreement and local assignment.
Honest parties agree on a $k$-ary tree of depth $L$, whose leaves are parties and whose internal nodes are size-$m$ committees with fewer than $m/4$ corrupted members.
A party may belong to multiple committees and learns its base committee, its internal committee memberships, and their members.
By Theorem~\ref{thm:quorumSize}, seed agreement produces an unpredictable string from which two public pseudorandom functions derive committee membership and leaf assignment.
More details can be found in Appendix~\ref{app:build-tree}.

\subsubsection{\InputShare{} with Domain Proof}

After locally comparing its vector coordinates to the current pivot, leaf $j$ secret-shares the resulting bit vector $b_{j,d} = \mathbf{1}[x_{j,d} \leq p_{t,d}]$ via \VSSShare{} to its assigned base committee.
A malicious leaf could share an arbitrary field element instead of a bit, biasing the aggregated count by an unbounded amount.
To prevent this, the committee verifies that each shared value $b$ satisfies $b \in \{0,1\}$ via the standard degree-$2$ domain proof: it runs $\Multiply([b], [1-b])$ and reconstructs the result using $\VSSReconst$.
If the reconstruction yields $0$, the share is accepted; otherwise the committee collectively replaces it with a canonical zero sharing.
This bounds the Byzantine contribution per iteration to $C_B \in [0, f]$.
Committee parties then locally sum the shares, producing a share of the partial count from the incoming leaves.

\subsubsection{Verifiable Share Redistribution via \Reshare}

\Reshare{} converts a $(\tau, m)$-Shamir sharing held by a child committee $\mathcal{C}_{\mathrm{ch}}$ into a fresh $(\tau, m)$-sharing held by its parent $\mathcal{C}_{\mathrm{pa}}$.
Two adversarial behaviors must be taken into account: a corrupted child may send sub-shares that are inconsistent with any degree-$\tau$ polynomial, or send consistent share but for the wrong secret.
Both are handled by $\VSSSubShare$~\cite[Prot.~6.8]{bgwSubshare}, ensuring that the parent committee receive correct sub-shares that correspond to the correct held secret.
Each parent reconstructs its share by combining the child committee’s contributions using Lagrange interpolation (see Fig.~\ref{fig:resharing}). Because every child contribution corresponds to an evaluation of the original polynomial, this weighted combination produces an evaluation of a newly randomized polynomial while preserving the same secret value at its constant term. Thus, the parent committee receives fresh shares of the value held by the child committee without revealing that value. The full protocol is given in Appendix~\ref{app:it-subprotocols}.

\begin{figure}[ht]
\centering
\begin{tikzpicture}[
    >={Latex[length=2mm]},
    old/.style={draw=orange!75, fill=orange!8, rounded corners=2pt,
                minimum width=1.15cm, minimum height=0.5cm, font=\large},
    lbl/.style={font=\footnotesize}
]

\matrix (M) [
    matrix of math nodes,
    nodes={minimum width=1.5cm, minimum height=0.54cm, anchor=center,
           font=\large, inner sep=1pt},
    row 2/.style={nodes={fill=blue!7}},
    column sep=0.6mm, row sep=0.6mm,
    left delimiter={[}, right delimiter={]}
]{
\tilde{\sigma}_{d,1\leftarrow1} & \tilde{\sigma}_{d,1\leftarrow2} & \cdots & \tilde{\sigma}_{d,1\leftarrow m} \\
\tilde{\sigma}_{d,r\leftarrow1} & \tilde{\sigma}_{d,r\leftarrow2} & \cdots & \tilde{\sigma}_{d,r\leftarrow m} \\
\vdots & \vdots & \ddots & \vdots \\
\tilde{\sigma}_{d,m\leftarrow1} & \tilde{\sigma}_{d,m\leftarrow2} & \cdots & \tilde{\sigma}_{d,m\leftarrow m} \\
};

\node[lbl, text=blue!70!black, left=3mm of M-1-1] {$P'_1$};
\node[lbl, text=blue!70!black, left=3mm of M-2-1] {$P'_r$};
\node[lbl, text=blue!70!black, left=3mm of M-4-1] {$P'_m$};

\node[font=\small, text=blue!70!black, rotate=90, anchor=center]
    at ([xshift=-9.5mm]M.west)
    {Parent $\mathcal C_{\mathrm{pa}}$ (new shares)};

\node[old] (s1) at ([yshift=6mm]M-1-1.north) {$\sigma_{d,1}$};
\node[old] (s2) at ([yshift=6mm]M-1-2.north) {$\sigma_{d,2}$};
\node[lbl]      at ([yshift=6mm]M-1-3.north) {$\cdots$};
\node[old] (sn) at ([yshift=6mm]M-1-4.north) {$\sigma_{d,m}$};

\node[lbl, text=orange!80!black, above=0.3mm of s1] {$P_1$};
\node[lbl, text=orange!80!black, above=0.3mm of s2] {$P_2$};
\node[lbl, text=orange!80!black, above=0.3mm of sn] {$P_m$};

\node[font=\small, lbl, text=orange!80!black] at ([yshift=13.5mm]M.north)
    {Child committee $\mathcal C_{\mathrm{ch}}$ (old shares)};

\node[
    draw=black!45,
    dashed,
    rounded corners=3pt,
    align=center,
    inner sep=4pt,
    anchor=north
]
(def) at ([yshift=-3.5mm]M.south)
{
    \footnotesize
    $\tilde{\sigma}_{d,r\leftarrow k}=g_{k,d}(\omega_r)$
    \enspace
    {\scriptsize\color{black!65}(child $k$ to parent $r$)}
};

\node[
    draw=blue!55,
    fill=blue!4,
    rounded corners=3pt,
    align=center,
    inner sep=4pt,
    anchor=north
]
(fresh) at ([yshift=-3mm]def.south)
{
    \footnotesize
    \textcolor{blue!70!black}{\textbf{Fresh share}}
    \enspace
    $\sigma'_{d,r}
    =
    \sum_{k=1}^{m}
    \lambda_k\,
    \tilde{\sigma}_{d,r\leftarrow k}$
};

\draw[->, thick]
    ([xshift=2mm]M.east |- M-2-4.center)
    to[out=0, in=0]
    ([xshift=1mm]fresh.east);

\end{tikzpicture}

\caption{
Matrix view of resharing.
Column $k$ contains the subshares from the child share $\sigma_{d,k}$.
Row $r$ contains all subshares received by parent party $P'_r$.
The new fresh share of $P'_r$ is obtained by recombining row $r$ as $\sigma'_{d,r}=\sum_{k=1}^{m}\lambda_k \tilde{\sigma}_{d,r\leftarrow k}$.
}
\label{fig:resharing}
\end{figure}

\subsubsection{Secure Noise Addition with \NoiseGen}
$\NoiseGen$ which uses only \Giskard{}'s existing BGW primitives: VSS, multiplication and robust reconstruction in order to generate DP noise in a distributed manner.
At a high level, $\NoiseGen$ has each committee member sample a discrete-Gaussian noise contribution bounded to $[-B,B-1]$, bit-decompose it, and uses \VSSShare{} on the bits, then the committee validates with a similar bit-verification protocol as \InputShare{} and replacing any malformed contribution with a zero sharing.
The surviving bits are then locally recombined so that, for each of the $N_{iter}$ queries, the parties hold a sharing $[Z_q]$ of aggregate noise that lies in the prescribed range and contains at least $m-\tau$ honest discrete-Gaussian samples.
We defer the full description to Protocol~\ref{alg:noisegen-dg} in Appendix~\ref{app:noisegen-dg}.

\subsection{Generalizing \Giskard{} to Other Robust Aggregation Rules}
\label{subsec:mpc-friendly}

The scalability of \Giskard{} does not require the aggregation rule itself to be associative, since even the coordinate-wise median is not associative.
Instead, its evaluation must reduce to a small number of public-query rounds whose answers are compact associative aggregates.
At round $t$, a public state $q_t$ defines local encodings $z_i^{(t)}=\phi_t(x_i;q_t)$ and summaries
\[
    S_t(I)=\bigoplus_{i\in I}z_i^{(t)},
    \qquad
    S_t(I\cup J)=S_t(I)\oplus_t S_t(J),
\]
where $\oplus_t$ is associative and the size of $S_t(I)$ is independent of $|I|$.
The global summary determines either the next public state or the final output.

We identify three requirements for a \emph{\Giskard-compatible realization}.
\begin{itemize}
    \item[\textbf{(G1)}] \textbf{Compact and query-efficient decomposability:} the rule admits compact mergeable summaries and a non-combinatorial number of public-query rounds.
    \item[\textbf{(G2)}] \textbf{Efficiently certifiable robustness:} Analogous to the computational surjectivity requirement of Franzese et al.~\cite{francez}, we could enforce that each contribution only admits a fixed admissible input, while robustness continues to hold for arbitrary Byzantine choices of such inputs.
    \item[\textbf{(G3)}] \textbf{Low secure arithmetic complexity:} producing and merging summaries uses mainly linear operations and few inexpensive multiplications, avoiding extensive comparisons, sorting, divisions, or spectral computations.
\end{itemize}

These requirements capture complementary bottlenecks: (G1) determines whether computation can be distributed over the committee hierarchy, (G2) prevents an unchecked Byzantine interface at the leaves, and (G3) controls the cost within each committee.
For example, the average function is perfectly decomposable and inexpensive to aggregate, but an unrestricted Byzantine summand has unbounded influence and therefore does not provide Byzantine robustness~\cite{krum}.

Table~\ref{tab:mpc-friendliness} illustrate the taxonomy.
Bounded output rules, using bits or signs, are the most MPC-friendly because they are cheap to certify and merge.
Examples include \Giskard{}'s binary-search coordinate-wise quantiles, $\ell_1$-RSA, DP-RSA~\cite{rsa,dprsa}, signSGD~\cite{signsgd}, and BRACE~\cite{brace}.

Input-dependent rules are less suitable.
RFA~\cite{rfa} admits a distributed Weiszfeld formulation but requires distance-dependent weights and divisions, while Krum~\cite{krum}, coordinate-wise median, and trimmed mean~\cite{trimmedmeans} require costly selection, comparison, or sorting.
SMEA~\cite{allouah_privacy-robustness-utility_2023} has mergeable statistics for fixed subsets, but searching $\binom{n}{f}$ candidates violates (G1).
CAF~\cite{allouahcafcor} supports compact summaries for weighted moments and covariance-vector products, yet repeated spectral estimation and nonlinear weight updates remain costly under MPC, violating (G3).

\begin{table}[t]
\centering
\caption{Compatibility of representative aggregation rules with the \Giskard{} interface.
A dash indicates that the rule is not analyzed under the $(f,\kappa)$-robustness definition.}
\label{tab:mpc-friendliness}
\footnotesize
\setlength{\tabcolsep}{3pt}
\renewcommand{\arraystretch}{1.05}
\begin{tabular}{@{}lcccc@{}}
\toprule
\textbf{Rule / realization}
& \textbf{$(f,\kappa)$-rob.}
& \textbf{G1}
& \textbf{G2}
& \textbf{G3} \\
$\ell_1$-RSA / DP-RSA~\cite{rsa,dprsa}
& -- & \cmark & \cmark & \cmark \\
signSGD-MV~\cite{signsgd}
& -- & \cmark & \cmark & \cmark \\
BRACE~\cite{brace}
& -- & \cmark & \cmark & \cmark \\
Geo.\ median / RFA~\cite{rfa}
& \cmark & \cmark & \xmark & \xmark \\
Krum~\cite{krum}
& \cmark & \xmark & \xmark & \xmark \\
CW median, direct~\cite{trimmedmeans}
& \cmark & \xmark & \xmark & \xmark \\
Trim.\ mean, direct~\cite{trimmedmeans}
& \cmark & \xmark & \xmark & \xmark \\
SMEA~\cite{allouah_privacy-robustness-utility_2023}
& \cmark & \xmark & \cmark & \xmark \\
CAF / CafCor~\cite{allouahcafcor}
& \cmark & \cmark & \cmark & \xmark \\
\midrule
\Giskard{} quantile
& \cmark & \cmark & \cmark & \cmark \\
\bottomrule
\end{tabular}
\end{table}

\section{Theoretical Analysis}
\label{sec:analysis}

\subsection{Security Analysis of Giskard}
\label{sec:security-statements}

We now state the main security results of \Giskard, the proofs being provided in the appendix in section \ref{app:giskard-security-proofs}. Our main argument is presented in a modular approach.

We first assume that $\BuildTree$ was successfully performed before running the secure median computation, so that all formed committees respect the byzantine tolerance bound.
We then establishes that all the defined sub-protocols are secure under the universal composability framework~\cite{canettiUC}. 
We finally show $\Giskard$ is perfectly secure using all the previous results and composition theorem~\cite{canettiUC}. 

We first collect the security guarantees of the underlying cryptographic building blocks.

\begin{lemma}
\label{lem:primitives-uc}
Protocols $\textsc{BA}$, $\VSSShare$, $\VSSReconst$, $\VSSSubShare$, and $\Multiply$ are UC-secure against static malicious adversaries.
\end{lemma}

We next state the security guarantees of the sub-protocols built from these components.

\begin{lemma}
\label{lem:subprotocols-uc}
Protocol $\InputShare{}$, $\Reshare$ and $\NoiseGen{}$ are UC-secure against static malicious adversaries.
\end{lemma}

The proofs are provided separately for each subprotocol respectively in Appendix \ref{app:giskard-proof}, \ref{app:reshare-proof} and \ref{app:noisegen-proof}. Combining these results through the UC composition theorem yields the security of the complete protocol.

\begin{theorem}
\label{thm:giskard-uc}
Assuming that $\BuildTree$ completes successfully, protocol $\Giskard$ securely realizes the approximated coordinate-wise median functionality in the UC framework against static malicious adversaries, provided that $f<n/4$.
\end{theorem}

The proof of Theorem~\ref{thm:giskard-uc} is provided in Appendix~\ref{app:giskard-proof}.

\subsection{Robustness Against Model Poisoning Attacks}
\label{subsec:robustness}
We analyze the median-search procedure of \S\ref{subsec:counting}, whose bit-domain proof guarantees $C_B(p_t)\in[0,f]$ at every iteration.

\begin{theorem}[Robustness of $\Giskard$]
\label{thm:giskard-robustness}
Let $\mathcal{H}\subseteq[n]$ be the honest set with $|\mathcal{H}|=n-f$ and $n>2f$.
Let $\bar{x}_{\mathcal{H}}$ and $\sigma_{\mathcal{H}}^2$ denote the mean and variance of the honest scalar inputs $x_j\in(-u,u]$.
For any adaptive Byzantine strategy, the output $\hat{M}$ of $\Giskard$ satisfies
\begin{equation}
\label{eq:scalar-robustness}
|\hat{M}-\bar{x}_{\mathcal{H}}|
\leq
\sqrt{\kappa_0}\,\sigma_{\mathcal{H}}+\frac{q}{2},
\qquad
\kappa_0
:=
\frac{n-f}{\lceil n/2\rceil-f}.
\end{equation}
If $q\leq\sigma_{\mathcal{H}}$, then
\begin{equation}
(\hat{M}-\bar{x}_{\mathcal{H}})^2
\leq
\left(
\kappa_0+\sqrt{\kappa_0}+\frac{1}{4}
\right)\sigma_{\mathcal{H}}^2.
\end{equation}
Applied coordinate-wise to inputs in $\mathbb{R}^D$, this gives
\begin{equation}
\label{eq:vector-robustness}
\begin{aligned}
\|\hat{M}-\bar{x}_{\mathcal{H}}\|^2
\leq {}&
2\kappa_0
\frac{1}{|\mathcal{H}|}
\sum_{j\in\mathcal{H}}
\|x_j-\bar{x}_{\mathcal{H}}\|^2
+\frac{Dq^2}{2}.
\end{aligned}
\end{equation}
\end{theorem}

Hence, $\Giskard$ tolerates any $f<n/2$ Byzantine inputs and satisfies Definition~\ref{def:robustagg} up to the additive discretization term $Dq^2/2$.
This term vanishes as $q\to0$ at a cost of $N_{\mathrm{iter}}=\lceil\log_2(2u/q)\rceil$ rounds.
Moreover, since $\kappa_0\leq2\left(1+\frac{f}{n-2f}\right)$, the multiplicative factor is in the same $\Theta(1)$ regime as CWMed, GM, and Krum~\cite{allouah_fixing_2023}.

The proof is deferred to Appendix~\ref{app:robustness-proof}.

\subsection{Differential Privacy Guarantees}
\label{subsec:dp-guarantee}

Let $K:=TdN_{\mathrm{iter}}$ and $h:=m-\tau$.
We use distributed discrete-Gaussian noise with
\begin{equation}
\sigma^2=\frac{K}{2\rho}, \qquad \sigma_c^2=\frac{\sigma^2}{h}, \qquad \eta_{\mathrm{DP}}=\sqrt{\frac{K\ln(2K/\beta)}{\rho}}.
\end{equation}
For $\sigma_c\geq 1/2$, define
\begin{equation}
\begin{aligned}
\vartheta_{\mathrm{DG}}
&=10\sum_{k=1}^{h-1}\exp\!\left(-2\pi^2\sigma_c^2\frac{k}{k+1}\right),\\
\xi_{\mathrm{DG}}
&=K\min\!\left\{\vartheta_{\mathrm{DG}},\frac{\vartheta_{\mathrm{DG}}}{\sqrt{h}\sigma_c}+\frac{\vartheta_{\mathrm{DG}}^2}{2}\right\}.
\end{aligned}
\end{equation}
By Proposition~13 of~\cite{kairouz2021distributed}, $\xi_{\mathrm{DG}}$ is the additional zCDP cost due to discrete-Gaussian aggregation.
For the truncation in \NoiseGen{}, let
\begin{equation}
q_B\leq 2\exp\!\left(-\frac{B^2}{2\sigma_c^2}\right), \qquad \beta_{\mathrm{tr}}:=\min\{1,Khq_B\}.
\end{equation}

\begin{theorem}[Privacy and Byzantine tolerance]
\label{thm:dp-property}
The ideal \Giskard{} mechanism is $(\rho+\xi_{\mathrm{DG}})$-zCDP.
Let $\bar{\rho}:=\rho+\xi_{\mathrm{DG}}$ and $\varepsilon_0:=\bar{\rho}+2\sqrt{\bar{\rho}\ln(1/\delta_0)}$ for any $\delta_0\in(0,1)$.
The truncated mechanism is
\begin{equation}
\left(\varepsilon_0,\delta_0+(1+e^{\varepsilon_0})\beta_{\mathrm{tr}}\right)\text{-DP}.
\end{equation}
With probability at least $1-\beta-\beta_{\mathrm{tr}}$, \NoiseGen{} incurs rank error at most $\eta_{\mathrm{DP}}$ simultaneously over all coordinates.
Hence, if the non-private protocol tolerates $f<n/2$ Byzantine inputs and $\tau<m/4$ corrupted root-committee parties, then
\begin{equation}
f_{\mathrm{DP}}=f-\eta_{\mathrm{DP}}-\tau B.
\end{equation}
\end{theorem}

\subsection{Committee Formation Guarantees}

We introduced the \BuildTree{} protocol, which includes the seed-agreement sub-protocol and the local committee-membership and leaf-assignment computations.
The benefit of this tree topology is most pronounced when committee size is small relative to $n$.
We show that such a tree exists under global corruption bounded by $(1/4 - \epsilon)n$:

\begin{theorem} [Minimum Required Committee Size]
\label{thm:treeConstruction}
Against a static adversary corrupting at most $(1/4 - \epsilon)n$ parties, $\BuildTree$ produces a $k$-ary tree of depth $L$ in which every committee has size $m = O(\epsilon^{-2} \ln n)$ and fewer than $m/4$ corrupted members, with probability $1 - n^{-c}$ for any constant $c$.
\end{theorem}

Full proof in Appendix~\ref{app:treeconstruction-proof}. 
The bound above is loose. Figure~\ref{fig:giskard_min_quorum_size} reports a tighter numerical evaluation at target $\Pr[\text{failure}] = 2^{-40}$, showing that the required committee size shrinks as the gap between the global corruption fraction and the $1/4$ per-committee requirement widens, and that for $f/n \in [0, 1/5]$ it remains stable as the network size grows.

\begin{figure}[ht]
\centering
\includegraphics[width=0.90\linewidth]{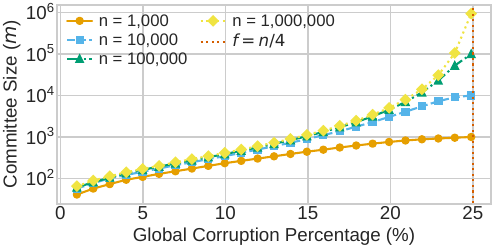}
\caption{Minimum committee size for varying global corruption fractions across different network sizes, at $\Pr[\text{failure}] = 2^{-40}$.}
\label{fig:giskard_min_quorum_size}
\end{figure}

\subsection{Communication Complexity}

Thanks to the topology of communication and the committee formation guarantees, we get the following 
\begin{theorem}[Communication Complexity of Giskard]
\label{thm:comm-bgw}
Let $N_{\mathrm{iter}}$ denote the number of median-search iterations, and let $d$ denote the model dimension.
Under the tree parameters of Theorem~\ref{thm:treeConstruction}, with $m=\Theta(\log n)$, constant branching factor $k=O(1)$, the following bounds hold:
\begin{enumerate}[label=(\roman*)]
  \item In the optimistic execution, in which no VSS complaint is invoked, the worst-case per-party communication cost is
  \begin{equation}
  \Delta_{\max}^{\mathsf{eff}}
  =
  O\!\left(dN_{\mathrm{iter}}\log^3 n\right),
  \mathsf{Total}
  =
  O\!\left(dN_{\mathrm{iter}}n\log^3 n\right).
  \end{equation}

  \item In the non-optimistic execution, complaints are disseminated using a reliable-broadcast protocol, where broadcasting $\ell$ field elements costs $O(m\ell)$ per party and $O(m^2\ell)$ in total. Then, the worst-case per-party communication cost is
    \begin{equation}
    \Delta_{\max}^{\mathsf{eff}}
    =
    O\!\left(dN_{\mathrm{iter}}\log^5 n\right),
    \mathsf{Total}
    =
    O\!\left(dN_{\mathrm{iter}}n\log^5 n\right).
    \end{equation}
\end{enumerate}
\end{theorem}

The proof is deferred to Appendix~\ref{app:comm-proof}.

We compare three BGW-based topologies for each median-search iteration using identical primitives: all-to-all (A2A)~\cite{ghavamipour_privacy-preserving_2024}, all-to-committee (A2C)~\cite{Jiang2025-xh,francez}, and \Giskard{}.
We report the worst-case per-party communication cost; for A2C, this is the cost of a committee member.
A2C and \Giskard{} use $m=O(\log n)$ committees sized via Chernoff and union bounds (non union bound needed for A2C), with $\tau<m/4$, while \Giskard{} uses constant branching factor $k$.
Table~\ref{tab:comm-comparison} summarizes the resulting costs.

\begin{table}[ht]
\centering
\caption{Maximum optimistic per-party and total communication per median-search iteration.}
\label{tab:comm-comparison}
\small
\begin{tabular}{lccc}
\toprule
& \textbf{Per-party cost} & \textbf{Total network cost} & \textbf{Rounds} \\
\midrule
A2A                & $O(dn^{2})$         & $O(dn^{3})$         & $O(1)$ \\
A2C               & $O(dn\log^{2} n)$   & $O(dn\log^{2} n)$   & $O(1)$ \\
\textbf{Giskard}   & $O(d\log^{3} n)$    & $O(dn\log^{3} n)$   & $O(\log n)$ \\
\bottomrule
\end{tabular}
\end{table}

\section{Experiments}
\label{sec:exp} 

We evaluate \Giskard{} on scalability, by comparing its communication and latency costs with MPC-based decentralized baselines (A2C and A2A), and model utility, by testing whether its binary-search median formulation remains robust under Byzantine attacks.
We implemented \Giskard{} in Rust using arkworks~\cite{arkworks}, and obtained plaintext baselines with a modified ByzFL framework in Python~\cite{gonzalez2025byzflresearchframeworkrobust}.
All experiments ran on an Ubuntu 24.04 workstation with an Intel Core i9-13950HX CPU and an NVIDIA RTX 5000 GPU. The source code is available in the \href{https://anonymous.4open.science/r/Giskard-ByzRobustPrivatelDL-922D}{Giskard Repository}.
Our evaluation is guided by the following research questions:

\begin{itemize}
\item \textbf{RQ1}: To what extent does \Giskard{} reduce per-party communication cost relative to A2A and A2C topologies?
\item \textbf{RQ2}: How does the network size $n$ affect the end-to-end latency of each topology?
\item \textbf{RQ3}:  Under which data-heterogeneity regimes does \Giskard{} retain robust performance under Byzantine attacks, and how does it compare with state-of-the-art plaintext robust aggregators?
\item \textbf{RQ4}: How does DP budget affect \Giskard{}'s resilience to Byzantine model-poisoning attacks ? 
\end{itemize}

\subsection{Comparing \Giskard{} per-party communication cost to A2A and A2C (RQ1)}

To isolate the effect of communication topology, we run the same median circuit under the same cost model and derive each subprotocol's cost by counting sent and received field elements.
For committee size $m$, we set $\tau=(m-1)/4$ and require fewer than $m/4$ Byzantine parties.
A2C chooses the smallest $m$ such that $\Pr[X\geq\lceil m/4\rceil]\leq\delta$, where $X\sim\mathrm{Hypergeom}(n,f,m)$ and $\delta=2^{-40}$.
\Giskard{} requires $C(k,L)=(k^L-1)/(k-1)$ committees, and as such we use the global union bound $C(k,L)\Pr[X\geq\lceil m/4\rceil]\leq\delta$.
We exhaustively search a feasible topology that minimizes \Giskard{}'s per-party cost over $2\leq k\leq32$ and $L$ values satisfying $k^{L-1}m\leq n$ and the union bound. Our MPC primitives are implemented in Rust using the arkworks library~\cite{arkworks}. From this implementation, we derive closed-form communication costs for $\VSSShare$, $\Multiply$, $\InputShare{}$, $\Reshare$ and $\NoiseGen$ by counting every field element sent, received, or broadcast. The resulting worst-case compositions per iteration are
\[
\begin{aligned}
C_{\mathrm{A2A}}
  &= C_{\mathrm{Input}}\!\left(n,\left\lfloor
     (n-1)/4\right\rfloor,n\right)
     +C_{\mathrm{NoiseGen}}(n),\\
C_{\mathrm{A2C}}
  &= C_{\mathrm{Input}}(m_{A2C},\tau_{A2C},n)
  +C_{\mathrm{NoiseGen}}(m) + C_{\mathrm{down}} ,\\
C_{\Giskard}
  &= C_{\mathrm{Input}}(m,\tau,n)
   +(L-1)kC_{\mathrm{Reshare}}(m,\tau)\\
  &\quad +C_{\mathrm{NoiseGen}}(m)+C_{\mathrm{down}}.
\end{aligned}
\]
with $C_{\mathrm{Input}}$ the \InputShare{} phase cost, $C_{\mathrm{Reshare}}$ the cost related to \Reshare{} between two committees, $C_{\mathrm{NoiseGen}}$ the \NoiseGen{} cost $C_{\mathrm{down}}$ for the broadcast down protocol. In the worst-case scenario, a party would have to assume all those prescribed roles. The model counts field-element, all other network related headers are excluded. 

Figure~\ref{fig:exp2} compares the analytical worst-case per-party communication of A2A, A2C, and \Giskard{}, and the resulting total communication cost. Panel~(a) represents total bytes versus $n$ at $f/n = 0.10$. Panel~(b) reports the cost as the network grows from \(n=10^3\) to \(10^6\), fixing \(f/n=0.10\), while Panel~(c) reports the A2C-to-\Giskard{} cost ratio for \(f/n\in\{0.01,0.05,0.10,0.20\}\). Each value includes all field elements sent and received during a scalar median computation on \(N_{\mathrm{iter}}=10\) iterations.

At $f/n=0.10$, A2C offers the lowest total cost, especially compared to A2A, whose cost skyrockets as $n$ increases. \Giskard{}, however, incurs a slight additional cost compared to A2C. 
\Giskard{}'s per-party cost grows from $392$~GB at $n=10^3$ to $1.99$~TB at $n=10^6$ at $f/n=0.10$. This is a $5.1{\times}$ increase for a $10^3{\times}$ larger network, compared with approximately $572{\times}$ for A2C and $1.00\cdot10^9{\times}$ for A2A.

At $n=10^6$, \Giskard{}'s per-party cost ranges from $146$~GB to $3.09$~PB across $f/n\in\{0.05,0.10,0.20\}$, compared with $45.3$~TB to $44.4$~PB for A2C. The A2A cost approximately $4.57\cdot10^{22}$ bytes regardless of byzantine fraction.
At $f/n = 0.01$, this corresponds to a $722\times$ reduction of \Giskard{} over A2C.
We explain this advantage by the fact that \Giskard{} distributes the MPC work across a tree of small committees, whereas A2C concentrates the load on a single committee and A2A scales committees with $n$. This make so that at $n=10^6$ and $f/n=0.10$, A2C's per-party cost is $99.7\%$ \InputShare{} because one committee validates all inputs, whereas \Giskard{} distributes this work across many committees, making the root commitee  \NoiseGen{} dominate ($89.5\%$) its much smaller per-party cost even thought \NoiseGen{} cost in both topologies are close ($1.78$~TB for \Giskard{} versus $0.78$~TB for A2C). 

\begin{takeaway}
    \textbf{Takeaway}: \Giskard{} greatly reduces per-party communication by distributing MPC across small committees, at the cost of a moderate increase in total communication compared with A2C.
\end{takeaway}

\begin{figure}[t]
  \centering
  \includegraphics[width=0.90\columnwidth]{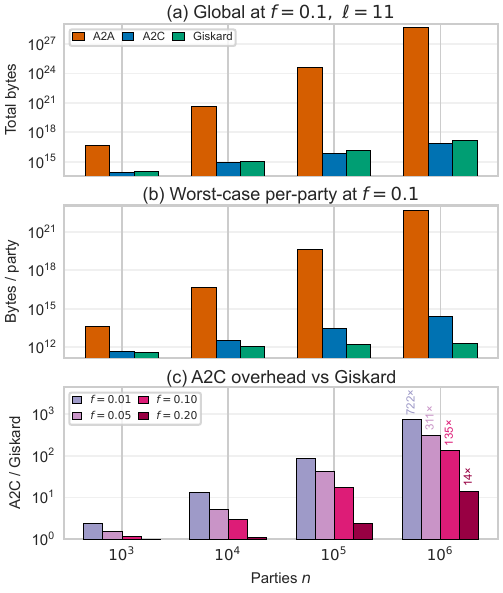}
  \caption{ (a)~Total communication costs $f/n = 0.1$. (b)~Worst-case Per-party communication cost at $f/n = 0.1$ for the three topologies. (c)~A2C overhead relative to \Giskard{} across four Byzantine fractions. \Giskard{} reduces per-party bandwidth by up to $722\times$ compared to A2C at $n = 10^6$ with $f/n = 0.01$, with the advantage growing as the network scales.}
  \label{fig:exp2}
\end{figure}

\subsection{Impact of Network Size on Total Cost and Latency (RQ2)}

In this experiment, we aim to study how each topology scales with the network size $n$  in end-to-end wall-clock time.
The results are presented in Figure~\ref{fig:exp4} for the corresponding minimum wall-clock time at $f/n = 0.10$ under a 1\,Gbps per-party bandwidth and a 1\,s round-trip delay, adding in time from executed local computation.
Each point in Figure~\ref{fig:exp4} represents the minimum end-to-end latency under the assumption that compute and transfer overlap within a round and same-layer committees in \Giskard{} execute concurrently, and that local computation is done with one thread, namely \InputShare{} is done sequentially.
At $f/n = 0.10$, our model places \Giskard{}'s wall-clock in the minutes-to-hours range across the tested $n$, A2C in the multi-hour to multi-day range from $n = 10^5$ on, and A2A beyond any practical horizon already at $n = 10^4$.
At small $n$ ($n = 10^2$) the three topologies fall within a small constant factor of each other, reflecting that topology choice matters only when $n \gg m$.

\begin{takeaway}
    \textbf{Takeaway:} \Giskard's enable a significant reduction in communication related runtime compare to the baseline topologies at scale.
\end{takeaway}

\begin{figure}[t]
  \centering
  \includegraphics[width=0.90\columnwidth]{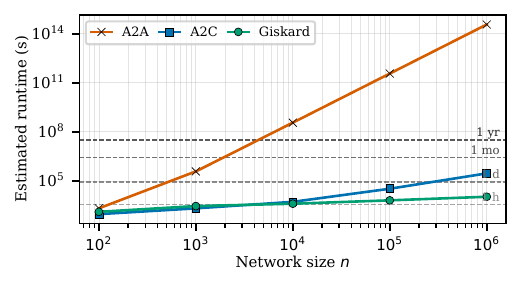}
  \caption{Implementation-calibrated end-to-end wall-clock time to completion versus network size $n$ at $f/n = 0.10$, under 1\,Gbps per-party bandwidth and 1\,s round-trip delay. Dashed lines mark sample durations.}
  \label{fig:exp4}
\end{figure}

\subsection{Comparing \Giskard{}'s Median to Baseline Aggregators under Poisoning (RQ3)}
\label{subsec:rq2}

In this experiment, we identify the data-heterogeneity regimes under which \Giskard{}'s MPC-tailored binary-search median retains robust model utility under Byzantine attacks. 

We compare it in plaintext with standard robust aggregation rules to isolate statistical robustness from cryptographic implementation costs. Using a fork of the ByzFL Python framework~\cite{gonzalez2025byzflresearchframeworkrobust}, we evaluate $n=100$ parties, of which $f=25$ are Byzantine. Table~\ref{tab:peak_accuracy_all} reports global test accuracy across several Dirichlet heterogeneity levels for \Giskard{}'s binary-search median (BSM), trimmed mean (TM), and CAF.
We consider two omniscient poisoning attacks, \textit{Fall of Empire} (FoE)~\cite{xie_fall_2019} and \textit{A Little is Enough} (ALIE)~\cite{baruch_little_2019}. 
We grant FoE and ALIE white-box access to the aggregation rule and all honest updates, optimizing their hyperparameters to maximize the aggregate’s distance from the honest mean. This is stronger than a Byzantine party’s view in \Giskard{}, where MPC hides honest updates, so we treat these attacks as \textit{conservative stress tests} to assess how much heterogeneity enable a power attacker to degrade the median.

The Dirichlet parameter $\alpha$ controls heterogeneity: smaller values yield concentrated, non-i.i.d.\ local data, while $\alpha=100$ approaches i.i.d. For MNIST at $\alpha=0.10$, $49\%$ of parties have at least $75\%$ of their samples from one class, versus none at $\alpha=1$. We include such severe stress regimes to identify when median-based aggregation fails, not as representative deployment assumptions.

Across MNIST and Fashion-MNIST, \Giskard{} remains competitive with TM until heterogeneity becomes severe. On MNIST, their FoE accuracies differ by at most $0.9\%$ for $\alpha \geq 0.5$, while \Giskard{} outperforms TM under ALIE. On Fashion-MNIST, the FoE gap grows from $3.4\%$ at $\alpha=0.75$ to $18.9\%$ at $\alpha=0.1$, although \Giskard{} remains stronger under ALIE throughout. This is notable because FoE specifically targets median and TM, though CAF remains the most robust at low $\alpha$. On CIFAR-10, the gap emerges at $\alpha=1.0$ and becomes substantially larger under FoE, while \Giskard{} withstands both attacks in the near-i.i.d.\ regime ($\alpha \geq 100$).
CAF is more robust in plaintext under severe heterogeneity but is not interchangeable with \Giskard{}: its repeated spectral estimation and nonlinear weight updates are too costly for MPC, violating our low-secure-arithmetic-complexity requirement (Section~\ref{subsec:mpc-friendly}). In contrast, \Giskard{} efficiently operates inside MPC and can add differential-privacy noise after secure aggregation, achieving close to central-DP utility without a trusted curator. Section~\ref{subsec:rq4} shows how this systems advantage can reverse the plaintext comparison.

\begin{table}[t]
  \caption{Peak test accuracy ($\uparrow$) under FoE and ALIE attacks under many dirichlet heterogeneity levels. *Training was terminated early because ALIE caused CAF to stall.}
  \label{tab:peak_accuracy_all}
  \centering
  \scriptsize
  \setlength{\tabcolsep}{2.5pt}
  \renewcommand{\arraystretch}{1.0}

  \begin{tabular}{r | rrr | rrr || rrr | rrr}
    \hline
    & \multicolumn{6}{c||}{\textsc{MNIST}} & \multicolumn{6}{c}{\textsc{FashionMNIST}} \\
    & \multicolumn{3}{c|}{FoE} & \multicolumn{3}{c||}{ALIE} & \multicolumn{3}{c|}{FoE} & \multicolumn{3}{c}{ALIE} \\
    $\alpha$ & \textbf{BSM} & TM & CAF & \textbf{BSM} & TM & CAF & \textbf{BSM} & TM & CAF & \textbf{BSM} & TM & CAF \\
    \hline
    0.10 & \textbf{0.18} & 0.47 & 0.97 & \textbf{0.84} & 0.76 & 0.98 & \textbf{0.15} & 0.33 & 0.74 & \textbf{0.59} & 0.53 & 0.81 \\
    0.50 & \textbf{0.93} & 0.93 & 0.98 & \textbf{0.97} & 0.97 & 0.99 & \textbf{0.58} & 0.62 & 0.78 & \textbf{0.73} & 0.68 & 0.82 \\
    0.75 & \textbf{0.95} & 0.96 & 0.98 & \textbf{0.97} & 0.97 & 0.99 & \textbf{0.67} & 0.70 & 0.79 & \textbf{0.74} & 0.74 & 0.83 \\
    1.00 & \textbf{0.96} & 0.97 & 0.99 & \textbf{0.98} & 0.97 & 0.99 & \textbf{0.71} & 0.71 & 0.78 & \textbf{0.77} & 0.75 & 0.83 \\
    100  & \textbf{0.97} & 0.97 & 0.99 & \textbf{0.98} & 0.98 & 0.99 & \textbf{0.77} & 0.78 & 0.82 & \textbf{0.82} & 0.81 & 0.84 \\
    \hline
  \end{tabular}

  \begin{tabular}{r | rrr | rrr}
    \hline
    & \multicolumn{6}{c}{\textsc{CIFAR10}} \\
    & \multicolumn{3}{c|}{FoE} & \multicolumn{3}{c}{ALIE} \\
    $\alpha$ & \textbf{BSM} & TM & CAF & \textbf{BSM} & TM & CAF \\
    \hline
    0.10 & \textbf{0.11} & 0.11 & 0.29 & \textbf{0.19} & 0.10 & 0.63* \\
    0.50 & \textbf{0.17} & 0.20 &  0.44  & \textbf{0.25} & 0.24 & 0.68*   \\
    0.75 & \textbf{0.20} & 0.26 &  0.64  & \textbf{0.29} & 0.26 & 0.67*   \\
    1.00 & \textbf{0.21} & 0.26 & 0.59   & \textbf{0.31} & 0.27 & 0.67*   \\
    100  & \textbf{0.57} & 0.48 & 0.70   & \textbf{0.48} & 0.43 & 0.70*   \\
    \hline
  \end{tabular}
\end{table}

\begin{takeaway}
    \textbf{Takeaway}: \Giskard{} remains competitive with robust plaintext aggregators across moderate heterogeneity, even under omniscient attacks. 
\end{takeaway}

\subsection{Impact of Differential Privacy on Byzantine Robustness (RQ4)}
\label{subsec:rq4}
We next evaluate whether \Giskard{} can provide DP while preserving Byzantine resilience. Because aggregation occurs inside MPC, noise is added \emph{after} secure aggregation while under Local DP, each update is perturbed beforehand, increasing the dispersion of honest updates. To isolate this statistical effect from MPC costs, we emulate both privacy mechanisms in plaintext while retaining their respective privacy accounting. For \Giskard{}, we compose zCDP additively across noisy count queries and convert the resulting guarantee to $(\epsilon,\delta)$-DP. For the Local DP baselines, we independently add Gaussian noise to each party's update before Median or CAF and account for the privacy budget using Opacus~\cite{opacus} and Renyi Differential Privacy~\cite{renyiDP}. We compare all mechanisms at the same final $(\epsilon,\delta)$ privacy guarantees, with $\delta =10^{-5}$.

Figure~\ref{fig:dp_comparison_exp} display the global test accuracy of model trained with \Giskard{} against Local DP baselines (using Median and CAF aggregators) on 12 settings encompassing 4 attacks and 3 privacy budgets.
On non-i.i.d EMNIST dataset with $\alpha = 1.0$ and $n= 11850$, \Giskard{} achieve the highest accuracy accross 11 of the 12 experiment settings. Under Sign Flipping, it remains near $91\%$ across all privacy budgets, outperforming both local-DP baselines by approximately $56\%$ at $\epsilon=10$ and $3\%$ at $\epsilon=100$.  Its advantage is largest under FoE: at $\epsilon=100$, \Giskard{} achieves $63.1\%$, versus $16.0\%$ for Local DP+CAF and $10\%$ accuracy for Local DP+Median. 
These results highlight \Giskard{}'s distributed-DP advantage. Although CAF is more robust under extreme plaintext heterogeneity, Local DP perturbs honest updates before aggregation. By enabling central-DP-style noise placement without a trusted curator, \Giskard{} avoids this distortion and can outperform Local DP even when combined with with CAF.

\begin{figure}[!ht]
\centering
\includegraphics[width=0.90\linewidth]{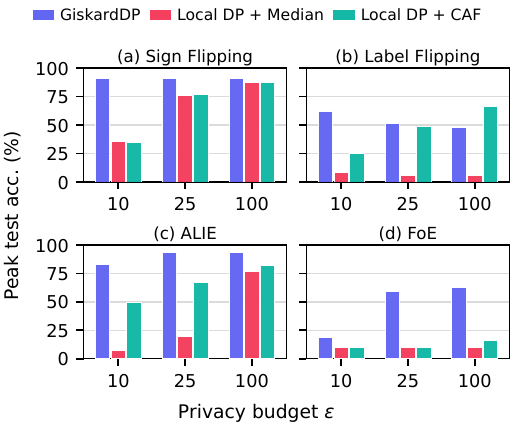}
\caption{Global test accuracy of \Giskard{} DP mechanism against Local DP baseline (Median and CAF Aggregators) on EMNIST with $n =11850$ under four Byzantine attacks, across privacy budgets with $f = 24\%$ at heterogeneity level $\alpha = 1.0$.} 
\label{fig:dp_comparison_exp}
\end{figure}

\begin{takeaway}
    \textbf{Takeaway}: \Giskard{} preserves substantially more utility under Byzantine attacks than Local DP. Secure post-aggregation noise enables strong privacy without perturbing individual honest updates before robust aggregation.
\end{takeaway}

\section{Related Work}
\label{sec:related}

Early privacy-preserving Byzantine-robust FL systems often combined separate mechanisms for robustness and privacy.
Biscotti~\cite{shayan_biscotti_2021} combines Multi-Krum and DP noise with secure aggregation, but performs robust aggregation in plaintext.
BREA~\cite{so2020byzantine} uses VSS-based MPC to compute pairwise distances, while the server receives the distance matrix in the clear to complete Multi-Krum.
Similarly, SHARE~\cite{velicheti2021secure} securely aggregates random client clusters but applies robust aggregation in plaintext to the resulting cluster means.
Thus, those approaches still expose intermediate information to the server, resulting in inherent privacy leakage.
A subsequent line of work focuses on combining robust aggregation with DP.
Allouah et al.~\cite{allouah_privacy-robustness-utility_2023} show that achieving a meaningful DP guarantee from local DP on top of $(f,\kappa)$-robust aggregation sharply degrades model utility.
CafCor~\cite{allouahcafcor} tackle this trade-off by replacing local DP with a correlated-noise mechanism across parties, and introducing a new $(f,\kappa)$-robust aggregator built on a covariance-bound-agnostic filter. It nonetheless operates in a classical FL setting with an honest-but-curious aggregation server.

A second line of work implements robust aggregation rules via cryptographic primitives in server-based settings. 
Sable~\cite{choffrut_sable} uses somewhat-homomorphic-encryption operators to compute trimmed-mean, but relies on a single powerful and honest-but-curious server. 
AlphaFL~\cite{Jiang2025-xh} and EIFFeL~\cite{eiffel} incorporate heuristic input validation into secure aggregation rather than formally robust rules under multi-server or client-server setups. 
Building on EIFFeL, DP-BREM+\cite{dpbrem-gu} combines zero-knowledge input validation with MPC-based distributed noise generation to provide DP guarantee. However, it remains restricted FL cross-silo topologies and exhibits per-client communication scaling linearly with model dimension, system size and adversary bounds. 

Removing the server assumption shifts robust aggregation assumptions entirely, forcing a stricter threat model.
SecureDL~\cite{ghavamipour_privacy-preserving_2024} use a maliciously secure MPC framework to implement a cosine-similarity filtering rule for aggregation.
However, its all-to-all topology scales poorly while cosine-similarity filtering offers no formal robustness guarantee.
Franzese et al.~\cite{francez} formalize robust serverless collaborative learning in an all-to-committee MPC setting, using a lightweight robust mechanism. Their single-layer topology concentrates computation on a small committee, giving per-party costs close to multi-server approaches, while the absence of DP leaves the global model vulnerable to privacy attacks. 
\Giskard{} addresses these gaps by evaluating a formally robust aggregator via hierarchical MPC committees and distributed DP, scaling robust decentralized learning to large networks with per-party sub-linear communication overhead.

\section{Conclusion}
\label{sec:conclusion}

In this work, we presented $\Giskard$, a scalable, robust, and privacy-preserving decentralized learning system that addresses the per-party communication bottleneck of existing approaches.
In fully decentralized architecture, $\Giskard$ achieves polylogarithmic per-party communication while combining Byzantine resilience and differential privacy in a malicious threat model.
We establish its security and robustness guarantees, analyze its communication and practical costs and empirically show that it is competitive at scale.
As future work, we plan to investigate non-interactive proof techniques to further reduce communication costs.

\begin{acks}
    
\end{acks}
This work was supported by the French government managed by the Agence Nationale de la Recherche (ANR) through France 2030 program with the reference ANR-23-PEIA-0005 (REDEEM project).
\begin{openscience}

We provide an artifact available in the following anonymized code repository using \href{https://anonymous.4open.science}{anonymous.4open.science}: \href{https://anonymous.4open.science/r/Giskard-ByzRobustPrivatelDL-922D}{Giskard Repository}\footnote{\url{https://anonymous.4open.science/r/Giskard-ByzRobustPrivatelDL-922D}}.

This artifact contains three elements:
\begin{itemize}
  \item A fork from the ByzFL repository to run the Collaborative Learning experiments under model poisoning attacks.
  \item A Rust PoC of the procotol using the arkworks library to run the communication cost experiment and compute estimated latency
  \item Additional python module to generate remaining figures
\end{itemize}

\end{openscience}

\begin{ai}
While writing this manuscript, AI tools (Gemini from Google, ChatGPT from OpenAI) were used for the following purposes: proofreading, improving text readability, editorial review, and searching additional related works results for comparison and proofs, all of which were read and verified by the authors.
Finally coding harness (Claude Code from Anthropic) provided assistance during the development of the protocol's PoC in Rust, which utilizes a well-known cryptography library (arkworks). Any contribution from the generative tool to the codebase were verified, tested, and compared against the description of the underlying protocols.

Apart from this, the intellectual contributions of this paper, namely the protocol design, the integration with the threat model, the security proofs, and the complexity analysis, were carried out without the assistance of an LLM. The authors take full responsibility for the accuracy, originality, and integrity of all content in the paper (text, figures, experimental data, and references).
\end{ai}

\bibliographystyle{IEEEtran}
\bibliography{bibliography}

\appendices

\begin{table}[h]
\centering
\caption{Notation.}
\label{tab:notation}
\small
\begin{tabular}{@{}ll@{}}
\toprule
\textbf{Symbol} & \textbf{Meaning} \\
\midrule
\multicolumn{2}{@{}l}{\textit{System and threat model}} \\
$n, \mathcal{P}, \mathcal{H}, \mathcal{A}$ & Parties; honest / adversarial sets \\
$f, \epsilon$ & Corruption count, $f = (1/4-\epsilon)n$ \\
\addlinespace
\multicolumn{2}{@{}l}{\textit{Tree and committees}} \\
$\mathcal{T}, k, L, \ell$ & Tree; branching, depth, level \\
$m, \tau$ & Committee size; threshold $\tau<m/4$ \\
$\mathcal{C}_{\ell,c}$ & Committee $c$ at level $\ell$ \\
$\mathcal{C}_j$ & Leaf $j$'s base committee \\
$\mathcal{C}_{\mathrm{ch}}, \mathcal{C}_{\mathrm{pa}}, \mathcal{C}_L$ & Child, parent, root committee \\
$P_i, P'_r$ & Parties in child / parent \\
\addlinespace
\multicolumn{2}{@{}l}{\textit{Median search}} \\
$D, u, q$ & Dimension, value range, precision \\
$\mathbf{x}_j, \mathbf{p}_t$ & Leaf input; pivot at iter.\ $t$ \\
$\mathsf{Left}, \mathsf{Right}$ & Binary-search bounds \\
$\mathbf{b}_j$ & Bits, $b_{j,d}=\mathbf{1}[x_{j,d}\leq p_{t,d}]$ \\
$\widetilde{\mathsf{count}}, \hat{M}$ & Aggregated noisy count; final aggregate \\
$N_{\mathrm{iter}}$ & Median Search iterations \\
\addlinespace
\multicolumn{2}{@{}l}{\textit{Secret sharing}} \\
$[\cdot], \sigma_i, \sigma'_r$ & Share: generic, child, parent \\
$f_d, \psi_i$ & Sharing / resharing polynomial \\
$\alpha_i, \omega_r, \lambda_i$ & Eval.\ points; Lagrange coeff. \\
$\mathcal{F}_{\textsc{X}}$ & Ideal functionality for $\textsc{X}$ \\
\bottomrule
\end{tabular}
\end{table}

\section{BGW Scheme Primitives}
\label{app:bgw-primitives}

\subsection{\VSSShare.}
We now describe the BGW verifiable secret sharing protocol as described in \cite[Prot. 5.6]{bgwSubshare}.
The protocol follows Shamir secret sharing, augmented with a bivariate commitment to enable verification. A dealer holding secret $s$ samples a random bivariate polynomial $B(x, y) \in \mathbb{F}_p[x, y]$ of degree $t-1$ in each variable, under the constraint $B(0, y) = q(y)$ where $q$ is the underlying univariate sharing polynomial with $q(0) = s$. The dealer sends each party $P_i$ the two univariate polynomials $f_i(x) = B(x, i)$ and $g_i(y) = B(i, y)$.
Every pair $(P_i, P_j)$ then performs a pairwise consistency exchange: $P_i$ sends $f_i(j)$ and $g_i(j)$ to $P_j$, which checks $f_i(j) \stackrel{?}{=} g_j(i)$ and $g_i(j) \stackrel{?}{=} f_j(i)$.
A failed check triggers a broadcast complaint. For each complaint from $P_j$, the dealer may either broadcast the full polynomial pair $(f_j, g_j)$ to resolve it publicly, or remain silent. Parties whose complaints were resolved broadcast \textsc{Consistent}. Each $P_i$ accepts $f_i(0) = B(0, i)$ as its share of $s$ if and only if at least $n - t$ \textsc{Consistent} messages are received.

\begin{functionality}{$\mathcal{F}_{\textsc{VSSShare}}$}

\textbf{Parameters.}
A committee $\mathcal C=\{P_1,\ldots,P_m\}$ with evaluation points $\alpha_1,\ldots,\alpha_m$.

\textbf{Input.}
A dealer holds a secret $s\in\mathbb F$.

\textbf{Functionality.}

\begin{enumerate}[leftmargin=1.8em,itemsep=4pt]

\item Sample a uniformly random polynomial $f$ such that $\deg(f)\leq\tau$ and $f(0)=s$.

\item Deliver $(\textsc{Share},\mathsf{sid},\mathsf{ssid},f(\alpha_i))$ to every $P_i\in\mathcal C$.

\end{enumerate}

The parties obtain a degree-$\tau$ sharing of $s$.

\end{functionality}

\begin{myprotocol}[label=alg:vss-share-bgw] {$\VSSShare$}

\textbf{Input.}
The dealer $D$ holds
\[
q(Z)\in\mathbb F[Z],
\qquad
\deg(q)\leq\tau.
\]

\textbf{Output.}
Each $P_i\in\mathcal C$ obtains
\[
\sigma_i=q(\alpha_i)
\]
or outputs $\bot$.

\vspace{3pt}\hrule\vspace{3pt}

\textbf{1. Distribution.}
The dealer samples a random bivariate polynomial $S(X,Y)$ of degree at most $\tau$ in each variable such that
\[
S(0,Y)=q(Y).
\]
For every $P_i\in\mathcal C$, the dealer privately sends
\[
f_i(X)=S(X,\alpha_i),
\qquad
g_i(Y)=S(\alpha_i,Y).
\]

\textbf{2. Cross-checking.}
Each $P_i$ sends
\[
\bigl(f_i(\alpha_j),g_i(\alpha_j)\bigr)
\]
to every $P_j\in\mathcal C$.
Upon receiving $(u_{j,i},v_{j,i})$ from $P_j$, party $P_i$ checks
\[
u_{j,i}=g_i(\alpha_j),
\qquad
v_{j,i}=f_i(\alpha_j).
\]
If either equality fails, $P_i$ broadcasts a complaint against $P_j$.
If no complaint is broadcast, $P_i$ outputs
\[
\sigma_i=f_i(0).
\]

\textbf{3. Complaint resolution.}
For each complaint, the dealer reveals the pair
\[
(f_j,g_j)
\]
of any party $P_j$ whose reported values are inconsistent with $S$.
Each party verifies every revealed pair using
\[
f_i(\alpha_j)=g_j(\alpha_i),
\qquad
g_i(\alpha_j)=f_j(\alpha_i).
\]
A party whose pair is revealed replaces its local pair accordingly.
Every other party that accepts all resolutions broadcasts $\textsc{Consistent}$.

\textbf{4. Output.}
If at least $m-\tau$ parties broadcast $\textsc{Consistent}$, then $P_i$ outputs
$\sigma_i=f_i(0)$.

Otherwise, $P_i$ outputs $\bot$.

\end{myprotocol}

\subsection{\Multiply.}
We now describe the BGW multiplication of two secrets protocol as described in \cite[Prot. A.1]{bgwSubshare}, in the specific case where $\tau < m/4$.
Given degree-$(t)$ sharings of $a$ and $b$, the protocol outputs a
degree-$(t)$ sharing of $ab$ in two steps:
\begin{enumerate}
\item Each party locally computes the product of its input shares, obtaining a
point on a degree-$2t$ polynomial that hides $ab$.
\item Each party subshares (i) its two original input shares and (ii) its
local product, each via a fresh random degree-$(t)$ polynomial. The
consistency of the input subshares certifies that the product subshare
indeed corresponds to the product of those inputs. Once verified, parties
locally compute a public linear combination of the product subshares to
recover a fresh degree-$(t)$ sharing of $ab$.
\end{enumerate}

\begin{functionality}{$\mathcal{F}_{\textsc{Multiply}}$}

\textbf{Parameters.}
A committee $\mathcal C=\{P_1,\ldots,P_m\}$ with evaluation points $\alpha_1,\ldots,\alpha_m$.

\textbf{Input.}
The committee holds two degree-$\tau$ sharings $a_i=f_a(\alpha_i)$ and $b_i=f_b(\alpha_i)$, where $\deg(f_a),\deg(f_b)\leq\tau$.

\textbf{Functionality.}

\begin{enumerate}[leftmargin=1.8em,itemsep=4pt]

\item Every party locally computes $c_i=f_a(\alpha_i)f_b(\alpha_i)$.

\item Invoke
\[
\mathcal F_{\VSSSubShare}\bigl(\{c_i\}_{P_i\in\mathcal C},\mathcal C\bigr).
\]
For every $i$, let $g_i$ denote the resulting degree-$\tau$ subsharing polynomial satisfying $g_i(0)=c_i$.

\item Let $\lambda_1,\ldots,\lambda_m$ be the public coefficients satisfying
\[
f_a(0)f_b(0)=\sum_{i=1}^{m}\lambda_i f_a(\alpha_i)f_b(\alpha_i).
\]
Define
\[
Q(X)=\sum_{i=1}^{m}\lambda_i g_i(X).
\]

\item Deliver $(\textsc{Share},\mathsf{sid},\mathsf{ssid},Q(\alpha_r))$ to every party $P_r$.

\end{enumerate}

Since
\[
Q(0)
=\sum_{i=1}^{m}\lambda_i g_i(0)
=\sum_{i=1}^{m}\lambda_i f_a(\alpha_i)f_b(\alpha_i)
=f_a(0)f_b(0),
\]
the outputs form a fresh degree-$\tau$ sharing of the product.
\end{functionality}

\begin{myprotocol}{$\Multiply$}

\textbf{Input.}
Each party $P_i$ holds $a_i=f_a(\alpha_i)$ and $b_i=f_b(\alpha_i)$.

\textbf{Output.}
Each party $P_i$ obtains $Q(\alpha_i)$, where $Q$ is a degree-$\tau$ sharing of $f_a(0)f_b(0)$.

\vspace{3pt}\hrule\vspace{3pt}

\textbf{1. Local multiplication.}
Each $P_i$ computes $c_i=a_i b_i$.

\textbf{2. Subsharing.}
Invoke $\VSSSubShare(c_i)$.
Each $P_i$ receives $C_1(\alpha_i),\ldots,C_n(\alpha_i)$, where $C_j(0)=c_j$.

\textbf{3. Local recombination.}
Let $\lambda_1,\ldots,\lambda_n$ be the first row of the inverse Vandermonde matrix.
Each $P_i$ computes
\[
Q(\alpha_i)=\sum_{j=1}^{n}\lambda_j C_j(\alpha_i)
\]
and outputs $Q(\alpha_i)$.

\end{myprotocol}

\subsection{\VSSReconst}
We now describe the BGW reconstruction of a secret protocol as described in \cite[Step 3 of Prot. 7.1]{bgwSubshare}.
Suppose that a secret $s \in \mathbb{F}$ is shared via Shamir Secret Sharing using a polynomial $f$ of degree $t$ such that $f(0) = s$, and that each party $P_i$ possesses a share $f(a_i)$.
Normally, if all participants are honest, they can broadcast their shares, and ultimately $t+1$ shares are sufficient to recover $s$ by interpolation.
However, we are in a malicious scenario, and in this case, we must guard against the possibility that some of the shares used to reconstruct the secret may be incorrect, which would skew the final result. Therefore, a mechanism is needed to efficiently identify and correct these incorrect shares.

We know here that all the shares contains atleast $n \geq 3t+1$ points lying in the same polynomial containing secret $s$. As such, it is possible to use an efficient Reed-Solomon decoding algorithm such as Berlekamp–Welch to recover the exact polynomial to get the secret .

\begin{functionality}{$\mathcal{F}_{\textsc{VSSReconst}}$}

\textbf{Parameters.}
A committee $\mathcal C=\{P_1,\ldots,P_m\}$ with evaluation points $\alpha_1,\ldots,\alpha_m$.

\textbf{Input.}
The committee holds a degree-$\tau$ sharing
\[
\sigma_i=f(\alpha_i),
\qquad
\deg(f)\leq\tau.
\]

\textbf{Functionality.}

\begin{enumerate}[leftmargin=1.8em,itemsep=4pt]

\item Recover the unique polynomial $f$ consistent with the honest parties' inputs.

\item Deliver
\[
(\textsc{Secret},\mathsf{sid},\mathsf{ssid},f(0))
\]
to every party.

\end{enumerate}

\end{functionality}

\begin{myprotocol}[label=alg:vss-reconst-bgw] {$\VSSReconst$}

\textbf{Input.}
Each $P_i\in\mathcal C$ holds a share $\beta_i=f(\alpha_i)$, where $\deg(f)\leq\tau$.

\textbf{Output.}
Each party obtains the shared value $s=f(0)$.

\vspace{3pt}\hrule\vspace{3pt}

\textbf{1. Broadcast shares.}
Each party $P_i$ broadcasts its share $\beta_i$.
$\widehat{\boldsymbol{\beta}}=(\widehat{\beta}_1,\ldots,\widehat{\beta}_m)$ denotes the vector of received values, which may contain incorrect shares.

\textbf{2. Reed--Solomon decoding.}
Each party locally Reed--Solomon decodes $\widehat{\boldsymbol{\beta}}$ and recovers the unique polynomial $f$ of degree at most $\tau$ whose evaluation vector agrees with $\widehat{\boldsymbol{\beta}}$ in all but at most $\tau$ positions.

\textbf{3. Reconstruction.}
Each party outputs
\[
s=f(0).
\]

\end{myprotocol}

\subsection{\VSSSubShare}
We now describe the BGW subsharing protocol as described in \cite[Prot. 6.8]{bgwSubshare}.
The \VSSSubShare{} protocol enables verifiable resharing of existing Shamir shares. Suppose a set of parties holds a secret s distributed into shares via a polynomial f, such that each party $P_i$ holds $b_i = f(a_i)$. Then each party $P_i$ generates a random polynomial of degree $t$, $g_i$, such that $g_i(0) = b_i = f(a_i)$. Here, shares are created based on the shares of the common secret. 

The challenge here is to figure out how to prevent a malicious party from sending evaluations corresponding to a polynomial whose constant term is $g_i(0) \neq f(a_i)$. 
To do this, we exploit the supermajority property of the honest parties and can therefore represent the vector of the constant parts of the generated polynomials as a Reed-Solomon codeword $\gamma = (g_1(0), \dots, g_n(0))$ that differs from the valid codeword $\beta = (f(a_1), \dots, f(a_n))$ by only $\tau$ values, making it possible to use a code correction algorithm. 
We can then write the difference between these two vectors as $ \gamma = \beta + e$, where $e$ is an error vector of maximum weight $\tau$. 
Let $H$ be the parity-check matrix of the corresponding Reed-Solomon code. 
Since $\beta$ is a valid codeword, we have $\beta H^T = 0$, and by multiplying the expression for $\gamma$ by $H^T$, we obtain the syndrome $s = \gamma H^T = (\beta + e) H^T = 0 + eH^T$.

The error-correcting algorithm can therefore recover this error vector from the syndrome, and thus directly identify the erroneous segments. In particular, for a vector $e$, all indices $i$ such that $e_i \neq 0$ correspond to the indices of the parts that have cheated, and whose evaluations must therefore be discarded and replaced, directly by setting $f(a_i) = g_i(0) - e_i$. It is therefore necessary to retrieve this syndrome $s$ in addition to the subshares, which is ensured securely by the $F_mat$ protocol. Specifically, after sending the subshares via \VSSShare, each party $i$ ends up with a subshare vector $(g_1(a_i), \dots, g_n(a_i))$. The parties i can then calculate locally, using the public matrix $H^T$, $Y(a_i) = (g_1(a_i), \dots, g_n(a_i)) H^T$, corresponding to the evaluation at $a_i$ of the vector of $2t$ polynomials of degree $t$, $Y(x)$, whose constant term is $Y(0) = (g_1(0), \dots, g_n(0)) H^T = \gamma H^T = s$. Revealing this polynomial thus allows the syndrome to be recovered; this is achieved by having the parties $i$ broadcast their evaluation of $Y(x)$ and using \VSSReconst{} to reconstruct the syndrome. 

Once the syndrome is obtained, the index $j$ of the cheating parties is identified, and for those parties, the corresponding subshares are broadcast to reconstruct $g_j(0)$, which in turn allows the correct original share $f(a_j) = g_j(0) - e_j$ to be recovered.

Note that in the version described in \cite{bgwSubshare}, the \VSSSubShare{} protocol takes place within a single committee. In our case, we use a slightly different variant, which takes place between two committees. In this case, the source committee will only subshare its shares as the dealer via \VSSShare{}, and it is up to the receiving committee to perform the syndrome calculation and error correction. Note that, given the properties of our committees (same size, same corruption bound), this modified protocol does not alter the algebraic properties of the sent shares in any way; therefore, the syndrome calculation remains the same as if the dealers were from the same committee, and thus the argument for correctness remains the

\begin{functionality}{$\mathcal{F}_{\textsc{VSSSubShare}}$}

\textbf{Parameters.}
A sending committee $\mathcal C_{\mathrm{ch}}=\{P_1,\ldots,P_m\}$ with evaluation points $\alpha_1,\ldots,\alpha_m$, and a receiving committee $\mathcal C_{\mathrm{pa}}=\{P'_1,\ldots,P'_m\}$ with evaluation points $\omega_1,\ldots,\omega_m$.

\textbf{Input.}
Each $P_i\in\mathcal C_{\mathrm{ch}}$ holds a value $\sigma_i$.

\textbf{Functionality.}

\begin{enumerate}[leftmargin=1.8em,itemsep=4pt]

\item For every $P_i\in\mathcal C_{\mathrm{ch}}$, sample a uniformly random polynomial $g_i$ such that $\deg(g_i)\leq\tau$ and $g_i(0)=\sigma_i$.

\item Deliver $(\textsc{Share},\mathsf{sid},\mathsf{ssid},g_i(\omega_r))$ to every $P'_r\in\mathcal C_{\mathrm{pa}}$.

\end{enumerate}

For every sender $P_i$, the receivers obtain a fresh degree-$\tau$ sharing of $\sigma_i$.

\end{functionality}

\begin{myprotocol}[label=alg:vss-subshare-bgw] {$\VSSSubShare$}

\textbf{Input.}
Each $P_i\in\mathcal C$ holds $\beta_i=f(\alpha_i)$, where $\deg(f)\leq\tau$.

\textbf{Output.}
Each $P_i$ obtains $g'_1(\alpha_i),\ldots,g'_m(\alpha_i)$, where each $g'_j$ has degree at most $\tau$ and $g'_j(0)=f(\alpha_j)$.

\vspace{3pt}\hrule\vspace{3pt}

\textbf{1. Subsharing.}
Each $P_i$ samples a random degree-$\tau$ polynomial $g_i$ satisfying $g_i(0)=\beta_i$.

\textbf{2. Consistency check.}
Let $H$ be the parity-check matrix of the Reed--Solomon code defined by $\alpha_1,\ldots,\alpha_m$.
The parties invoke $\mathcal F_{\mathrm{mat}}^{H}$ on $g_1,\ldots,g_m$.
Each $P_i$ obtains $g_1(\alpha_i),\ldots,g_m(\alpha_i)$ and the syndrome
\[
\mathbf s=(g_1(0),\ldots,g_m(0))H^{T}.
\]
Each party decodes $\mathbf s$ with an error correction algorithm to obtain the error vector $\mathbf e=(e_1,\ldots,e_m)$.

\textbf{3. Correction.}
For every $k$ with $e_k=0$, each $P_i$ sets $g'_k(\alpha_i)=g_k(\alpha_i)$.
For every $k$ with $e_k\neq0$, the parties exchange their values $g_k(\alpha_i)$, reconstruct $g_k$, and set
\[
g'_k(\alpha_i)=g_k(0)-e_k=f(\alpha_k).
\]

Each $P_i$ outputs $g_i$ and $g'_1(\alpha_i),\ldots,g'_m(\alpha_i)$.

\end{myprotocol}

\section{Giskard Sub-Protocols Description}
\label{app:it-subprotocols}

This appendix details the sub-protocols that instantiate Giskard with the primitives are from the BGW scheme~\cite{bgwSubshare}, we recall only what is needed. We present here the sub-protocols used in \Giskard{}: \BuildTree{} ,\InputShare{}, \Reshare{}, \NoiseGen{} and \BroadcastDown.

\subsection{\BuildTree{}}
\label{app:build-tree}

We call the protocol \BuildTree, that is done locally by each parties, assuming that the Byzantine Aggreement phase worked, and that they all got the same random seed, they then sample the committee membership and Leaf assignment.

\textbf{Committee membership.}
For each tree level $\ell \in \{1, \ldots, L\}$, we derive a permutation $\pi_\ell$ of $[n]$ by Fisher--Yates shuffling seeded with $H(\mathsf{seed} \,\|\, \ell)$, where $H$ is a public hash function modeled as a random oracle.
The committees at level $\ell$ are the $k^{L-\ell}$ disjoint $m$-blocks in the prefix of $\pi_\ell$.
By construction, every party sits in at most one committee per level, so in the worst case a party participates in $L$ committees, one per level.
Within any level, each committee is marginally a uniformly random $m$-subset of $[n]$, so per-committee corruption follows $\mathrm{Hyp}(n, f, m)$; the global failure probability is bounded by a union bound over all committees.

\textbf{Leaf assignment.}
Each party $i \in [n]$ is assigned to leaf position $\pi_L(i)$.
The base committee of party $i$ is the level-$(L-1)$ committee whose
$k$ children include leaf $\pi_L(i)$, i.e., committee
$\lceil \pi_L(i) / k \rceil$ in the level-$(L-1)$ ordering.
Since $\pi_L$ is a permutation, every party occupies exactly one leaf
and every leaf is occupied by exactly one party.

\subsection{\InputShare{}}

After locally comparing its vector coordinates to the current pivot, leaf $j$ secret-shares the resulting bit vector $b_{j,d} = \mathbf{1}[x_{j,d} \leq p_{t,d}]$ via \VSSShare to its assigned base committee.
A malicious leaf could share an arbitrary field element instead of a bit, biasing the aggregated count by an unbounded amount.
To prevent this, the committee verifies that each shared value $b$ satisfies $b \in \{0,1\}$ via the standard degree-$2$ domain proof: it runs $\Multiply([b], [1-b])$ and reconstructs the result using $\VSSReconst$.
If the reconstruction yields $0$, the share is accepted; otherwise the committee collectively replaces it with a canonical zero sharing.
This bounds the Byzantine contribution per iteration to $C_B \in [0, f]$.
Committee parties then locally sum the validated shares, jointly produce a shares of 0 to be added, then by linearity of Shamir sharing, this produces a fresh share of the partial count from the incoming leaves, which becomes the input to the tree-aggregation phase. It is to be noted that the share refresh is necessary only if all the leaves are byzantines, on the other case the produce share of the partial count is already uniformly random, from the presence of atleast one honest leaf.

\begin{myprotocol}[label=alg:inputshare-bgw]
{$\InputShare{}$}

\textbf{Common input.}
The parties hold a description of $\mathbb F$, distinct nonzero evaluation points $\{\alpha_i\}_{P_i\in\mathcal C_{1,c}}$, the public pivot $\mathbf p_t$, and the identifiers $(\mathsf{sid},\mathsf{ssid})$.

\textbf{Aiding ideal-functionality initialization.}
Upon invocation, the trusted parties computing $\VSSShare$, $\Multiply$, and $\VSSReconst$ receive the static corrupted committee set $I_c$.
The honest parties do not receive $I_c$.

\textbf{Input.}
Each honest leaf $j\in\mathcal L_c\setminus\mathcal B_c$ holds $\mathbf x_j\in[u]^D$.
A Byzantine leaf may submit arbitrary values in $\mathbb F$.
The committee $\mathcal C_{1,c}$ has size $m$ and threshold $\tau<m/3$.

\textbf{Output.}
Each $P_i\in\mathcal C_{1,c}$ obtains $[\sigma_{c,d}]_i$ for every $d\in[D]$, where
\[
\begin{aligned}
\sigma_{c,d}
&=
\sum_{j\in\mathcal L_c}\widetilde b_{j,d},
\\
\widetilde b_{j,d}
&=
\begin{cases}
\mathbf 1[x_{j,d}\leq p_{t,d}], & j\notin\mathcal B_c,\\
b_{j,d}, & b_{j,d}\in\{0,1\},\\
0, & \text{otherwise}.
\end{cases}
\end{aligned}
\]

\vspace{3pt}\hrule\vspace{3pt}

\textbf{1. Share and validate.}
For every $j\in\mathcal L_c$ and $d\in[D]$, an honest leaf sets $b_{j,d}=\mathbf 1[x_{j,d}\leq p_{t,d}]$, while a Byzantine leaf may choose any $b_{j,d}\in\mathbb F$, and invokes $\VSSShare(b_{j,d})$.
A missing or rejected sharing is replaced by $[0]$.
The committee computes
\[
[e_{j,d}]
\gets
\Multiply([b_{j,d}],[1]-[b_{j,d}]),
\]
\[
e_{j,d}
\gets
\VSSReconst([e_{j,d}]),
\]
and sets
\[
[\widetilde b_{j,d}]
\gets
\begin{cases}
[b_{j,d}], & e_{j,d}=0,\\
[0], & e_{j,d}\neq0.
\end{cases}
\]

\textbf{2. Aggregate.}
For every $d\in[D]$, each $P_i$ locally computes
\[
[\widehat\sigma_{c,d}]_i
\gets
\sum_{j\in\mathcal L_c}[\widetilde b_{j,d}]_i.
\]

\textbf{3. Refresh.}
For every $d\in[D]$, the committee invokes
\[
[z_{c,d}]
\gets
\Multiply([0],[0]).
\]
Here $[0]$ denotes the public degree-$\tau$ sharing whose shares are all zero.
By the definition of the BGW multiplication functionality, $[z_{c,d}]$ is a fresh uniformly random degree-$\tau$ sharing of zero, conditioned on the output shares selected by the corrupted committee members.
Each $P_i$ outputs
\[
[\sigma_{c,d}]_i
\gets
[\widehat\sigma_{c,d}]_i+[z_{c,d}]_i.
\]
\end{myprotocol}

\subsection{\Reshare{}}

\Reshare{} converts a $(\tau, m)$-Shamir sharing held by a child committee $\mathcal{C}_{\mathrm{ch}}$ into a fresh $(\tau, m)$-sharing held by its parent $\mathcal{C}_{\mathrm{pa}}$ (all committees have size $m$ and threshold $\tau < m/4$).
Two adversarial behaviors must be taken into account: a corrupted child may send sub-shares that are inconsistent with any degree-$\tau$ polynomial, breaking reconstruction. Even consistent sub-shares must provably interpolate to the partial sum held by the child committee, not an adversarially chosen value.

Both are handled by $\VSSSubShare$~\cite[Prot.~6.8]{bgwSubshare}.
Each $P_i \in \mathcal{C}_{\mathrm{ch}}$ holds a share $\sigma_i = f(\alpha_i)$ of the child's partial sum, where $f$ is the degree-$\tau$ sharing polynomial and $\alpha_i$ is $P_i$'s evaluation point.
$P_i$ samples a fresh degree-$\tau$ polynomial $g_i$ with $g_i(0) = \sigma_i$ and sends $g_i(\omega_r)$ to each parent $P'_r \in \mathcal{C}_{\mathrm{pa}}$ at that party's evaluation point $\omega_r$.
The functionality guarantees $g_i(0) = f(\alpha_i)$ via Reed-Solomon decoding over the honest children's sub-shares.
Each parent then locally computes
\begin{equation}
\sigma'_r = \sum_{i \in \mathcal{C}_{\mathrm{ch}}} \lambda_i \cdot g_i(\omega_r),
\label{eq:reshare-recombine}
\end{equation}
where $\{\lambda_i\}$ are the Lagrange coefficients interpolating $f(0)$ from $\{f(\alpha_i)\}_{i \in \mathcal{C}_{\mathrm{ch}}}$.
Since $g_i(0) = f(\alpha_i)$, equation~\eqref{eq:reshare-recombine} evaluates a fresh degree-$\tau$ polynomial at $\omega_r$ whose constant term is $\sum_i \lambda_i f(\alpha_i) = f(0)$, the value shared by the child committee.
We describe the protocol for a single coordinate as it runs in parallel across all $d \in [D]$.

\begin{myprotocol}[label=alg:reshare-bgw]{$\textsc{Reshare}$}

\textbf{Input.} Each $P_i\in\mathcal{C}_{\mathrm{ch}}$ holds \[ \sigma_i=f(\alpha_i), \qquad \deg(f)\leq\tau, \] where $f(0)$ is the shared secret. 

\textbf{Output.} Each $P'_r\in\mathcal{C}_{\mathrm{pa}}$ obtains \[ \sigma'_r=f'(\omega_r), \] where $f'$ is a degree-$\tau$ polynomial satisfying $f'(0)=f(0)$.

\vspace{3pt}\hrule\vspace{3pt} \textbf{1. Subsharing.} 
Invoke \[ \VSSSubShare \bigl( \{\sigma_i\}_{P_i\in\mathcal{C}_{\mathrm{ch}}}, \mathcal{C}_{\mathrm{pa}} \bigr). \] 
For every sender $P_i$, each receiver $P'_r$ obtains $g_i(\omega_r)$, where \[ \deg(g_i)\leq\tau \qquad\text{and}\qquad g_i(0)=\sigma_i=f(\alpha_i). \] 

\textbf{2. Local recombination.} Let $\lambda_1,\ldots,\lambda_m$ be the public Lagrange coefficients for interpolation at zero over $\alpha_1,\ldots,\alpha_m$. 
Each $P'_r\in\mathcal{C}_{\mathrm{pa}}$ locally computes \[ \sigma'_r\gets\sum_{i=1}^{m}\lambda_i g_i(\omega_r) \] and outputs $\sigma'_r$.
\end{myprotocol}

\subsection{\BroadcastDown{}}

We now present a simple protocol, exploiting the guarantees of a well formed tree, to reliably send down a public value through the tree. It exploits the honest supermajority assumption, where each parties from a parent committee send down the common public value to the parties of a child committee. For parties of that child committee, they locally agree if among all the received value, one have a multiplicity of over $m/4$. Under the assumption of synchronous communication, authenticated point-to-point channels
\begin{myprotocol}[label=alg:downward-broadcast]{$\textsc{Broadcast Down (in the Tree)}$}
\textbf{Input.}
Each $P'_r \in \mathcal{C}_{\mathrm{pa}}$ holds the same public value
$v \in \mathbb{R}$ (e.g., a reconstructed pivot coordinate).
Both committees have size $m$ with threshold $\tau < m/4$.

\textbf{Output.}
Each honest $P_i \in \mathcal{C}_{\mathrm{ch}}$ holds $v$.
\vspace{3pt}\hrule\vspace{3pt}
\textbf{1. Send.}
Each $P'_r \in \mathcal{C}_{\mathrm{pa}}$ sends $v$ to every $P_i \in \mathcal{C}_{\mathrm{ch}}$ over authenticated channels.

\textbf{2. Majority vote.}
At the end of the prescribed synchronous round, each $P_i\in\mathcal C_{\mathrm{ch}}$ replaces every missing message by $\bot$.
It outputs the unique value with multiplicity greater than $m/4$, or $\bot$ if no such value exists.
\end{myprotocol}

\subsection{\NoiseGen{}}
\label{app:noisegen-dg}

Protocol~\ref{alg:noisegen-dg} details the DP noise-generation procedure.
Each party contributes an independent bounded discrete-Gaussian sample, so the aggregate noise $Z_q$ contains at least $m-\tau$ honest samples regardless of adversarial behavior.
The bit decomposition and the multiplicative check $[\beta][1-\beta]$ enforce that every accepted contribution is a well-formed integer in $[-B,B-1]$; a contribution that fails validation is replaced by $[0]$ rather than aborting, which preserves liveness at the cost of the corrupted parties' noise only.
The truncation to $[-B,B-1]$ introduces a statistical deviation from $\mathcal N_{\mathbb Z}(0,\sigma_c^2)$ of at most $2\exp(-B^2/2\sigma_c^2)$ per sample, which is negligible for $\ell \ge \log_2(\sigma_c) + \omega(\log\lambda)$.
The choice $p>\max(m,2(N+mB))$ guarantees that the sum $\sum_i \zeta_i^{(q)}$, added to any query result of magnitude at most $N$, does not wrap around modulo $p$.

\begin{myprotocol}[label=alg:noisegen-dg]{$\NoiseGen$}
\textbf{Input.}
Privacy budget $\rho$, query count $N_{iter}$,
\[
\sigma_c=\sqrt{\frac{N_{iter}}{2\rho(m-\tau)}},
\qquad
B=2^{\ell-1},
\]

\[
p>\max\!\bigl(m,2(N+mB)\bigr),
\]
and a committee $\mathcal C$ of size $m$ with $\tau<m/4$.
\textbf{Output.}
For each $q\in[N_{iter}]$, a sharing
\[
[Z_q]=\sum_{i=1}^{m}[\zeta_i^{(q)}],
\]
where every honest $\zeta_i^{(q)}$ follows
$\mathcal N_{\mathbb Z}(0,\sigma_c^2)$ conditioned on $[-B,B-1]$, and every
accepted contribution lies in that interval.
\vspace{3pt}\hrule\vspace{3pt}
\textbf{1. Sampling and sharing.}
Each $P_i$ samples $\zeta_i^{(q)}$ from
$\mathcal N_{\mathbb Z}(0,\sigma_c^2)$ until
$\zeta_i^{(q)}\in[-B,B-1]$, writes
\[
\zeta_i^{(q)}+B=\sum_{j=0}^{\ell-1}2^j\beta_{i,j}^{(q)},
\]
and VSS-shares every bit $\beta_{i,j}^{(q)}$.

\textbf{2. Validation.}
For every bit, the committee opens
\[
[\beta_{i,j}^{(q)}][1-\beta_{i,j}^{(q)}].
\]
If a sharing fails or any opened value is nonzero, the corresponding
contribution is replaced by $[0]$.

\textbf{3. Aggregate.}
For every accepted contribution, the committee locally computes
\[
[\widetilde\zeta_i^{(q)}]
=
\sum_{j=0}^{\ell-1}2^j[\beta_{i,j}^{(q)}]-B.
\]
For every rejected contribution, it sets
\[
[\widetilde\zeta_i^{(q)}]=[0].
\]
The committee then locally computes
\[
[Z_q]
=
\sum_{i=1}^{m}[\widetilde\zeta_i^{(q)}].
\]
\end{myprotocol}

\section{Security Analysis Proofs}
\label{app:giskard-security-proofs}
\subsection{Security Proof Methodology}
\label{app:modular-proof-methodology}

In this section, we describe the definitions, results, and proof techniques we will use to prove the security of our cryptographic protocols. We primarily rely on results from the work from \cite{bgwSubshare}, \cite{canettiUC}, and \cite{klrStraightLineUC}. We first the perfect security of our protocol in the standalone model \cite{goldreich2004foundations}, then use lifting theorem \cite[Theorem 1.2]{klrStraightLineUC} to push it into UC security \cite{canettiUC}.

\paragraph{Execution and corruption model.}
Here, we reiterate what has already been described in the Section~\ref{subsec:system-model} on the system model. 
In particular, regarding the execution of MPC protocols, the committee sizes $m$ and $\tau$ and the maximum corruption threshold are defined.
The attack model assumes static malicious adversaries $\mathcal A$, implying a fixed set $|I|\leq \tau$.
We assume the existence of authenticated private channels for secure point-to-point communication.

\paragraph{Real and ideal executions}
In the real execution, all parties run protocol $\Pi$, while the adversary $\mathcal A$ have additional input $z$ and controls the parties indexed by $I$.
$\mathsf{REAL}_{\Pi,\mathcal A(z),I}(\mathbf x)$ is the random variable consisting of the adversary's view (the full transcript of communication visible to him, together with the outputs of all honest parties.

In the ideal execution, call to the protocol are replace to call to a trusted party computing the functionality $\mathcal F_{\Pi}$, while an ideal model simulator $\mathcal S$ controls the corrupted parties and interacts with $\mathcal F_{\Pi}$.
$\mathsf{IDEAL}_{\mathcal F,\mathcal S(z),I}(\mathbf x)$ is the random variable consisting of the simulator's output together with the outputs delivered by $\mathcal F_{\Pi}$ to the honest parties.

\begin{definition}[Perfect stand-alone security]
A protocol $\Pi$ perfectly realizes a functionality $\mathcal F_{\Pi}$ against static malicious corruptions of at most $\tau$ parties if, for every real-model adversary $\mathcal A$, there exists an polynomial time ideal-model simulator $\mathcal S$ in that of $\mathcal A$ such that, for every corruption set $I\subseteq[m]$ with $|I|\leq\tau$, every input vector $\mathbf x$, and every auxiliary input $z$,
\[
\mathsf{REAL}_{\Pi,\mathcal A(z),I}(\mathbf x)
\equiv
\mathsf{IDEAL}_{\mathcal F_{\Pi},\mathcal S(z),I}(\mathbf x),
\]
where $\equiv$ denotes equality of distributions. 
\end{definition}

\paragraph{Hybrid execution} In a protocol $\Pi$ running a $(\mathcal F_1,\ldots,\mathcal F_k)$-hybrid execution, parties use the protocol as intended but invoke trusted parties computing the functionalities $(\mathcal F_1,\ldots,\mathcal F_k)$. The adversary $\mathcal A$ is still there with additional input $z$, controlling the parties indexed by $I$.
When calling one of those functionalities, a unique sub-session identifier $ssid$ is created to separate that instance from the other invocations.
$\mathsf{HYBRID}^{\mathcal F_1,\ldots,\mathcal F_k}_{\Pi,\mathcal A(z),I}(\mathbf x)$  is the random variable consisting of the adversary's view (the full transcript of communication visible to him, together with the outputs of all honest parties). For a protocol $\secProtocol$ running a $(\mathcal F_1,\ldots,\mathcal F_k)$-hybrid execution we denote $
\rwInstantiation{\secProtocol}{\secSubProtocol_1, \dots, \secSubProtocol_k}$ to the protocol that runs $\secProtocol$ but for all $i \in \{1,\dots,k\}$, it runs sub-protocol $\secSubProtocol_i$ instead of calling ideal functionality $\mathcal F_i$.

\begin{definition}[Perfect security in the hybrid model]

A protocol $\Pi$ perfectly realizes a functionality $\mathcal F_{\Pi}$ against static malicious corruptions of at most $\tau$ parties if, for every $(\mathcal F_1,\ldots,\mathcal F_k)$-hybrid model adversary $\mathcal A$, there exists an polynomial time ideal-model simulator $\mathcal S$ in that of $\mathcal A$ such that, for every corruption set $I\subseteq[m]$ with $|I|\leq\tau$, every input vector $\mathbf x$, and every auxiliary input $z$,
\[
\mathsf{HYBRID}^{\mathcal F_1,\ldots,\mathcal F_k}_{\Pi,\mathcal A(z),I}(\mathbf x)
\equiv
\mathsf{IDEAL}_{\mathcal F_{\Pi},\mathcal S(z),I}(\mathbf x).
\]
\end{definition}
Let $\Pi$ be a protocol intended to realize a functionality $\mathcal F_{\Pi}$ calling ideal functionalities $\mathcal F_1,\ldots,\mathcal F_k$.
For every static malicious adversary $\mathcal A$, we construct an ideal-model simulator $\mathcal S$ such that, for every corruption set $I$ with $|I|\leq\tau$, every input vector $\mathbf x$, and every auxiliary input $z$,

The simulator internally runs $\mathcal{A}$, emulating honest parties and functionalities with canonical defaults for invalid messages. It forwards extracted inputs to $\mathcal F_{\Pi}$ and fixes corrupted output shares while $\mathcal F_{\Pi}$ conditionally samples honest shares.
If the simulator accesses $\mathcal A$ only through its input-output interface and never inspects, rewinds, or restarts it, it is straight-line and black-box.

\paragraph{Lifting to UC security.}
After proving (in the standalone model) the perfect security of a protocol in the hybrid model, if the used simulator present the property of being straight-line, and blackbox, we can use the \cite[Theorem 1.2]{klrStraightLineUC} lifting theorem which state that every protocol that is perfectly secure in the stand-alone model, and has a straight-line black-box simulator, is secure under concurrent general composition, equivalent to (static) universal composability. It allow to prove that a protocol $\Pi$ UC-realizes $\mathcal{F}_{\Pi}$ in the hybrid model. Furthermore, if the ideal functionalities from the hybrid model are also UC-realized by some existing protocols, we can use the UC composition theorem~\cite{canettiUC}, actually removing the ideal functionalities called by the protocol one, to states that the composition $\Pi$ calling those sub-protocols UC-realizes $\mathcal{F}_{\Pi}$.

\subsection{Existing UC Security Results}
\label{app:subprotocols-uc-proof}
\begin{proof}
    Using the BGW full security proof in \cite{bgwSubshare} already shows that $\VSSShare$, $\VSSReconst$, $\VSSSubShare$, $\Multiply$ and $\VSSReconst$ are perfectly secure using a straight-line black-box simulator. Using \cite[Theorem 1.2]{klrStraightLineUC}, it proves that those protocols are UC-Secure.
\end{proof}

\subsection{UC Security of \Reshare{}}
\label{app:reshare-proof}

\subsubsection{Ideal Functionality}

We now present the ideal functionality $\mathcal F_{\Reshare}$ in Func.~\ref{func:f-reshare}, it is to be noted that it is equivalent to the Renew-Share functionality from \cite{SecureShuffling}.

\begin{functionality}{$\mathcal{F}_{\Reshare}$} \label{func:f-reshare} 

\textbf{Parameters.} A sending committee $\mathcal{C}_{\mathrm{ch}}=\{P_1,\ldots,P_m\}$ with evaluation points $\alpha_1,\ldots,\alpha_m$, and a receiving committee $\mathcal{C}_{\mathrm{pa}}=\{P'_1,\ldots,P'_m\}$ with evaluation points $\omega_1,\ldots,\omega_m$.

\textbf{Input.} The sending committee holds evaluations \[ \sigma_i=f(\alpha_i), \qquad \deg(f)\leq\tau, \] where $f(0)$ is the shared secret. 

\textbf{Aiding ideal-functionality initialization.}
At the start of the execution, the corruption-aware functionality $\mathcal F_{\VSSSubShare}$ receives the static corrupted sets $I_{\mathrm{ch}}\subseteq\mathcal C_{\mathrm{ch}}$ and $I_{\mathrm{pa}}\subseteq\mathcal C_{\mathrm{pa}}$.
The honest parties do not receive these sets.

\textbf{Functionality.} 
\begin{enumerate}[leftmargin=1.8em,itemsep=4pt]

\item Invoke 
\[ \mathcal{F}_{\VSSSubShare} \bigl( \{\sigma_i\}_{P_i\in\mathcal{C}_{\mathrm{ch}}}, \mathcal{C}_{\mathrm{pa}} \bigr). \] 
For every $i$, let $g_i$ be the resulting degree-$\tau$ subsharing polynomial satisfying 
\[ g_i(0)=\sigma_i=f(\alpha_i). \]

\item Let $\lambda_1,\ldots,\lambda_m$ be the public Lagrange coefficients satisfying 
\[ f(0)=\sum_{i=1}^{m}\lambda_i f(\alpha_i). \] 
Define \[ f'(X)=\sum_{i=1}^{m}\lambda_i g_i(X). \]

\item Deliver 
\[ \bigl(\textsc{Share},\mathsf{sid},\mathsf{ssid},f'(\omega_r)\bigr) \] 
to every $P'_r\in\mathcal{C}_{\mathrm{pa}}$.

\end{enumerate} 

Since \[ f'(0)=\sum_{i=1}^{m}\lambda_i g_i(0) =\sum_{i=1}^{m}\lambda_i f(\alpha_i) =f(0), \] the outputs form a degree-$\tau$ sharing of the original secret. 

\end{functionality}

\subsubsection{Security Analysis}

We can now prove the following result:

\begin{lemma}[Security of the \Reshare{} protocol]
    Protocol \Reshare{} securely realizes $\mathcal{F}_{\Reshare}$ under the universal composability framework against static malicious adversaries for $f < m /4$
\end{lemma}

\begin{proof}
We first present the simulator $\mathcal S_{\Reshare}$ in the ideal world, then prove that the transcript it generates is equivalent to that of the real execution.

The simulator $\mathcal S_{\Reshare}$ emulates the interface of $\mathcal F_{\VSSSubShare}$ for $\mathcal A$ by forwarding the corrupted senders' inputs to $\mathcal F_{\Reshare}$ and giving $\mathcal A$ the subshares and leakage returned by $\mathcal F_{\Reshare}$.

For every corrupted receiver $P'_r$, the simulator obtains
\[
g_1(\omega_r),\ldots,g_m(\omega_r)
\]
from $\mathcal F_{\Reshare}$ and gives these values to $\mathcal A$ as the output of the simulated $\mathcal F_{\VSSSubShare}$ call.
It may locally compute
\[
f'(\omega_r)
=
\sum_{i=1}^{m}\lambda_i g_i(\omega_r),
\]
although this value is already a deterministic function of $\mathcal A$'s view.

The simulator outputs whatever $\mathcal A$ outputs and halts.
No additional messages or randomness are simulated because the recombination step is entirely local.

We first establish perfect security in the $\mathcal F_{\VSSSubShare}$-hybrid model.
Let $\mathcal A$ be a static malicious adversary and let $\mathcal S_{\Reshare}$ be the simulator described above.

For each receiver $P'_r$, let
\[
\mathbf g_r
=
\bigl(g_1(\omega_r),\ldots,g_m(\omega_r)\bigr)
\]
denote its output from $\mathcal F_{\VSSSubShare}$.
Define the deterministic map
\[
L_r(\mathbf g_r)
=
\sum_{i=1}^{m}\lambda_i g_i(\omega_r).
\]
The only local operation performed after the $\mathcal F_{\VSSSubShare}$ call is for each honest receiver $P'_r$ to output $L_r(\mathbf g_r)$.
This operation sends no messages and uses no randomness.

By construction, $\mathcal F_{\Reshare}$ exposes to the simulator exactly the subshares and leakage that corrupted parties receive from $\mathcal F_{\VSSSubShare}$.
The simulator therefore generates the adversary's hybrid-model view exactly.
Furthermore, for every honest receiver $P'_r$, both the protocol and $\mathcal F_{\Reshare}$ output the same deterministic value $L_r(\mathbf g_r)$.

Consequently,
\[
\left\{
\mathsf{HYBRID}^{\mathcal F_{\VSSSubShare}}_{\Reshare,\mathcal A(z),I}(\boldsymbol\sigma)
\right\}
\equiv
\left\{
\mathsf{IDEAL}_{\mathcal F_{\Reshare},\mathcal S_{\Reshare}(z),I}(\boldsymbol\sigma)
\right\}.
\]
Thus, $\Reshare{}$ perfectly realizes $\mathcal F_{\Reshare}$ in the $\mathcal F_{\VSSSubShare}$-hybrid model.

The simulator invokes $\mathcal A$ once, treats it as a black box, forwards messages online, and never rewinds it.
It is therefore a straight-line black-box simulator of complexity polynomial in that of $\mathcal A$.
The perfect-security lifting theorem implies that $\Reshare{}$ perfectly UC-realizes $\mathcal F_{\Reshare}$ in the hybrid model. We then prove UC Security of $\InputShare{}$ using the UC composition theorem~\cite{canettiUC}.
\end{proof}

\subsection{UC Security of \InputShare{}}
\label{app:inputshare-proof}

\subsubsection{Ideal functionality}
\label{app:inputshare-functionality}

We now present the ideal functionality $\mathcal F_{\InputShare}$ in Func.~\ref{func:f-inputshare}.

\begin{functionality}{$\mathcal F_{\InputShare{}}$}
\label{func:f-inputshare}

\textbf{Parameters.}
Leaves $\mathcal L_c$; committee $\mathcal C_{1,c}$;
Byzantine leaves $\mathcal B_c$; corrupted committee members
$I_c\subseteq\mathcal C_{1,c}$ with
$|I_c|\leq \tau<m/3$; public pivot
$\mathbf p_t\in\mathbb R^D$.

\textbf{Initialization.}
Upon invocation, $\mathcal F_{\InputShare{}}$  receives the static Byzantine-leaf set $\mathcal B_c$ and the static corrupted-committee set $I_c$, where  $|I_c|\leq \tau$.
Neither set is revealed to the honest parties.

\medskip
\textbf{Functionality.}

\begin{enumerate}[leftmargin=1.8em,itemsep=4pt]

\item For every honest $j\in\mathcal L_c\setminus\mathcal B_c$, receive
\[
(\textsc{Input},\mathsf{sid},\mathsf{ssid},\mathbf x_j)
\]
and set
\[
\widetilde b_{j,d}
=
\mathbf 1[x_{j,d}\leq p_{t,d}]
\qquad(d\in[D]).
\]
Reveal to $\mathcal S$ only that $j$ submitted an input.

For every Byzantine $j\in\mathcal B_c$, receive
\[
\{\beta_{j,d}\}_{d\in[D]}
\in(\mathbb F\cup\{\bot\})^D
\]
from $\mathcal S$ and set
\[
\widetilde b_{j,d}
=
\begin{cases}
\beta_{j,d}, & \beta_{j,d}\in\{0,1\},\\
0,           & \beta_{j,d}\notin\{0,1\}\text{ or }\beta_{j,d}=\bot.
\end{cases}
\]

\item For every $d\in[D]$, set
\[
\sigma_{c,d}
=
\sum_{j\in\mathcal L_c}\widetilde b_{j,d}
\in\mathbb F.
\]

\item Receive
\[
(\textsc{CorruptedShares},\mathsf{sid},\mathsf{ssid},
\{\rho_{i,d}\}_{P_i\in I_c,d\in[D]})
\]
from the ideal adversary before the prescribed corrupted-share deadline.
For every missing or malformed $\rho_{i,d}$, set
\[
\rho_{i,d}=0.
\]
For every $d\in[D]$, independently choose a uniformly random polynomial
\[
f_{c,d}\in\mathbb F[X]
\]
of degree at most $\tau$ satisfying
\[
f_{c,d}(0)=\sigma_{c,d}
\]
and
\[
f_{c,d}(\alpha_i)=\rho_{i,d}
\qquad(P_i\in I_c).
\]

\item Immediately send
\[
(\textsc{Shares},\mathsf{sid},\mathsf{ssid},
\{f_{c,d}(\alpha_i)\}_{d\in[D]})
\]
to every honest $P_i\in\mathcal C_{1,c}\setminus I_c$.
\end{enumerate}
\end{functionality}

\subsubsection{Security Analysis}
\label{app:inputshare-simulator}

We now prove the following result: 

\begin{lemma}[Security of the \InputShare{} protocol]
    Protocol \InputShare{} securely realizes $\mathcal{F}_{\InputShare}$ under the universal composability framework against static malicious adversaries for $f < m /4$
\end{lemma}

\begin{proof}
We first present the simulator $\mathcal S_{\InputShare{}}$ in the ideal world, then prove that the transcript it generates is equivalent to that of the real execution. 

The simulator $\mathcal S_{\InputShare{}}$ interacts externally with $\mathcal F_{\InputShare{}}$ and internally invokes the hybrid-model adversary $\mathcal A$.
It internally emulates the honest leaves, the honest committee members, and all calls to $\VSSShare$, $\Multiply$, and $\VSSReconst$.

\begin{enumerate}[leftmargin=1.8em,itemsep=4pt]

\item
The simulator invokes $\mathcal A$ with its auxiliary input.

\item
For every honest leaf $j\notin\mathcal B_c$ and coordinate $d\in[D]$, the simulator chooses an independent uniformly random polynomial $q^{(0)}_{j,d}\in\mathbb F[X]$ of degree at most $\tau$ satisfying $q^{(0)}_{j,d}(0)=0.$

It internally simulates the corresponding $\VSSShare$ invocation and gives $\mathcal A$ the shares $\{q^{(0)}_{j,d}(\alpha_i)\}_{P_i\in I_c}$ 
that the corrupted committee members expect to receive.

\item
For every Byzantine leaf $j\in\mathcal B_c$ and coordinate $d\in[D]$, the simulator internally obtains from $\mathcal A$ the polynomial or value submitted to $\VSSShare$.
If the invocation fixes a valid degree-$\tau$ sharing, the simulator defines $\beta_{j,d}$ to be its constant term.
If the invocation is missing or rejected, the simulator sets $\beta_{j,d}=\bot$.
The simulator externally sends the resulting values $\{\beta_{j,d}\}_{j\in\mathcal B_c,d\in[D]}$ to $\mathcal F_{\InputShare{}}$.

\item
For every honest leaf and coordinate, the simulator internally emulates the multiplication and reconstruction calls using the simulated sharings of $0$ and $1$.
It simulates the reconstructed validation value as $0$.

\item
For every Byzantine leaf and coordinate, let
\[
\overline\beta_{j,d}
=
\begin{cases}
\beta_{j,d}, & \beta_{j,d}\in\mathbb F,\\
0, & \beta_{j,d}=\bot.
\end{cases}
\]
The simulator internally emulates the multiplication and reconstruction calls so that the reconstructed value is
\[
e_{j,d}
=
\overline\beta_{j,d}(1-\overline\beta_{j,d}).
\]
It retains the simulated sharing precisely when $e_{j,d}=0$, and otherwise replaces it by the public sharing $[0]$.

\item
For every $d\in[D]$, the simulator locally aggregates the shares held by the corrupted committee members and obtains
\[
\{[\widehat\sigma_{c,d}]_i\}_{P_i\in I_c}.
\]

\item
For every $d\in[D]$, the simulator internally emulates the call
\[
[z_{c,d}]
\gets
\Multiply([0],[0]).
\]
Let $\eta_{i,d}$ be the corrupted output share fixed during this simulated multiplication call.
The simulator defines
\[
\rho_{i,d}
=
[\widehat\sigma_{c,d}]_i+\eta_{i,d}
\]
for every $P_i\in I_c$.

\item
The simulator externally sends
\[
(\textsc{CorruptedShares},\mathsf{sid},\mathsf{ssid},
\{\rho_{i,d}\}_{P_i\in I_c,d\in[D]})
\]
to $\mathcal F_{\InputShare{}}$.

\item
The simulator outputs whatever $\mathcal A$ outputs and halts.

\end{enumerate}

Now we prove that the hybrid execution with adversary $\mathcal A$ is identically distributed to the ideal execution with $\mathcal S_{\InputShare{}}$.
$\mathcal A$ is a  static malicious adversary controlling the Byzantine leaves $\mathcal B_c$ and the committee members in $I_c$, where $|I_c|\leq \tau$.

For every honest leaf $j\notin\mathcal B_c$ and coordinate $d\in[D]$, let $q_{j,d}$ denote the uniformly random degree-$\tau$ polynomial used in the hybrid execution, so that $q_{j,d}(0)=b_{j,d}$.

In the ideal execution, $\mathcal S_{\InputShare{}}$ instead chooses a uniformly random degree-$\tau$ polynomial $q^{(0)}_{j,d}$ satisfying $q^{(0)}_{j,d}(0)=0$.

Since $|I_c|\leq\tau$, perfect privacy of Shamir sharing implies
\[
\{q_{j,d}(\alpha_i)\}_{P_i\in I_c}
\equiv
\{q^{(0)}_{j,d}(\alpha_i)\}_{P_i\in I_c}.
\]
Independence of the sharing polynomials extends this equality jointly to all honest leaves and coordinates.

Since $[1-b_{j,d}]=[1]-[b_{j,d}]$, the corrupted parties' shares of both inputs to the validation multiplication are jointly identically distributed in the hybrid and ideal executions.
Moreover, for every honest leaf, $b_{j,d}(1-b_{j,d})=0$, whereas the simulator uses sharings of $0$ and $1$, whose product also has constant term $0$.

Hence, by the specifications of $\Multiply$ and $\VSSReconst$, the multiplication and reconstruction views of the corrupted parties are identically distributed in the two executions.

For every Byzantine leaf and coordinate, the simulator internally invokes the same hybrid functionalities with $\mathcal A$ and therefore obtains the same value $\beta_{j,d}$ fixed by the corresponding VSS invocation.
The reconstructed validation value is
\[
e_{j,d}
=
\overline\beta_{j,d}(1-\overline\beta_{j,d}).
\]
Since $\mathbb F$ is a field,
\[
e_{j,d}=0
\quad\Longleftrightarrow\quad
\overline\beta_{j,d}\in\{0,1\}.
\]
Consequently, the protocol and $\mathcal F_{\InputShare{}}$ use the same effective value $\widetilde b_{j,d}$.

It follows that the corrupted parties' aggregate shares
\[
\{[\widehat\sigma_{c,d}]_i\}_{P_i\in I_c}
\]
are identically distributed in the hybrid and ideal executions.

Fix a coordinate $d$, and let $\widehat f_{c,d}$ denote the degree-$\tau$ polynomial defined by the aggregate sharing before refresh in the hybrid execution.
Its constant term is
\[
\widehat f_{c,d}(0)=\sigma_{c,d}.
\]

In the refresh call, $\mathcal{F}_{\Multiply}$ samples a polynomial $Z_{c,d}$ uniformly from the set
\[
\left\{
Z\in\mathbb F[X]:
\deg(Z)\leq \tau,
\ Z(0)=0,
\ Z(\alpha_i)=\eta_{i,d}\ \text{for every }P_i\in I_c
\right\},
\]
where the values $\{\eta_{i,d}\}_{P_i\in I_c}$ are fixed by the corrupted parties during the hybrid call.
The final sharing polynomial is
\[
f_{c,d}
=
\widehat f_{c,d}+Z_{c,d}.
\]
Its corrupted evaluations satisfy
\[
f_{c,d}(\alpha_i)
=
[\widehat\sigma_{c,d}]_i+\eta_{i,d}
=
\rho_{i,d}.
\]
For fixed $\widehat f_{c,d}$, the map
\[
Z\longmapsto \widehat f_{c,d}+Z
\]
is a bijection from the preceding set of zero-sharing polynomials to
\[
\left\{
f\in\mathbb F[X]:
\deg(f)\leq \tau,
\ f(0)=\sigma_{c,d},
\ f(\alpha_i)=\rho_{i,d}\ \text{for every }P_i\in I_c
\right\}.
\]
Therefore, conditioned on the adversary's view and the corrupted output shares $\{\rho_{i,d}\}_{P_i\in I_c}$, the final polynomial in the hybrid execution is uniform over exactly the set used by $\mathcal F_{\InputShare{}}$.
The refresh calls use independent randomness for distinct coordinates, so the same conclusion holds jointly for every $d\in[D]$.

The ideal functionality computes $\sigma_{c,d}$ from the actual honest inputs and the effective Byzantine inputs.
Since the values $\{\rho_{i,d}\}_{P_i\in I_c}$ are identically distributed in the hybrid and ideal executions, and conditioned on these values the honest output shares have exactly the same distribution, we obtain
\begin{align*}
\left\{
\mathsf{HYBRID}_{\InputShare{},\mathcal A(z),I_c}(\mathbf x)
\right\}
\equiv 
\left\{ 
\mathsf{IDEAL}_{\mathcal F_{\InputShare{}},\mathcal S_{\InputShare{}}(z),I_c}(\mathbf x)
\right\}.
\end{align*}

Thus, $\InputShare{}$ perfectly realizes $\mathcal F_{\InputShare{}}$ in the ($\VSSShare$, $\Multiply$, $\VSSReconst$)-hybrid model.

The simulator $\mathcal S_{\InputShare{}}$ is straight-line and black-box as it invokes $\mathcal A$ once, processes all messages in their prescribed order, maintains separate state for every $\mathsf{ssid}$, and never rewinds. Since we established perfect security, we can use \cite[Theorem 1.2]{klrStraightLineUC} lifting theorem to prove to prove that $\InputShare{}$ UC-realizes $\mathcal{F}_{\InputShare{}}$ in the $(\VSSShare$, $\Multiply$, $\VSSReconst)$-hybrid model against static malicious adversaries.
\end{proof}

\subsection{UC Security of \NoiseGen{}}

\subsubsection{Ideal Functionality}

We now present the ideal functionality $\mathcal F_{\NoiseGen}$ in Func.~\ref{func:f-noisegen}.

\begin{functionality}{$\mathcal F_{\NoiseGen}$}
\label{func:f-noisegen}

\textbf{Parameters.}
A committee $\mathcal C=\{P_1,\ldots,P_m\}$ with distinct nonzero evaluation points $\alpha_1,\ldots,\alpha_m$, a static corrupted set $I\subseteq\mathcal C$ satisfying $|I|\leq\tau<m/4$, a query count $N_{iter}$, a truncation bound $B=2^{\ell-1}$, and a public parameter $\sigma_c$.

\textbf{Functionality.}
\begin{enumerate}[leftmargin=1.8em,itemsep=4pt]

\item
For every honest $P_i\notin I$ and query $q\in[N_{iter}]$, independently  (through rejection sampling)
\[
\zeta_i^{(q)}
\sample
\mathcal N_{\mathbb Z}(0,\sigma_c^2)
\mid
[-B,B-1].
\]
Let
\[
\zeta_i^{(q)}+B
=
\sum_{j=0}^{\ell-1}2^j\beta_{i,j}^{(q)}
\]
be its unique bit decomposition.
Reveal to the ideal adversary only that $P_i$ submitted its contribution.

\item
For every corrupted $P_i\in I$ and query $q\in[N_{iter}]$, receive
\[
\{\widehat\beta_{i,j}^{(q)}\}_{j=0}^{\ell-1}
\in
(\mathbb F\cup\{\bot\})^\ell
\]
from the ideal adversary.
Set
\[
\widetilde\zeta_i^{(q)}
=
\begin{cases}
\displaystyle\sum_{j=0}^{\ell-1}2^j\widehat\beta_{i,j}^{(q)}-B,
& \widehat\beta_{i,j}^{(q)}\in\{0,1\}\text{ for every }j,\\
0,
& \text{otherwise}.
\end{cases}
\]
For every honest $P_i\notin I$, set
\[
\widetilde\zeta_i^{(q)}=\zeta_i^{(q)}.
\]

\item
For every $q\in[N_{iter}]$, compute
\[
Z_q
=
\sum_{i=1}^{m}\widetilde\zeta_i^{(q)}
\in\mathbb F.
\]

\item
Receive
\[
(\textsc{CorruptedShares},\mathsf{sid},\mathsf{ssid},
\{\rho_{r,q}\}_{P_r\in I,q\in[N_{iter}]})
\]
from the ideal adversary before the prescribed deadline.
Replace every missing or malformed $\rho_{r,q}$ by $0$.

\item
For every $q\in[N_{iter}]$, independently sample a uniformly random polynomial $F_q\in\mathbb F[X]$ of degree at most $\tau$ satisfying
\[
F_q(0)=Z_q
\]
and
\[
F_q(\alpha_r)=\rho_{r,q}
\qquad(P_r\in I).
\]

\item
For every honest $P_r\notin I$, deliver
\[
(\textsc{Share},\mathsf{sid},\mathsf{ssid},
\{F_q(\alpha_r)\}_{q\in[N_{iter}]})
\]
to $P_r$.
The corrupted output shares are delivered directly to the ideal adversary.

\end{enumerate}
\end{functionality}

\subsubsection{Security Analysis}
\label{app:noisegen-proof}

We now prove the following lemma: 

\begin{lemma}[Security of the \NoiseGen{} protocol]
    Protocol \NoiseGen{} securely realizes $\mathcal{F}_{\NoiseGen}$ under the universal composability framework against static malicious adversaries for $f < m /4$
\end{lemma}

\begin{proof}

We first construct the simulator $\mathcal S_{\NoiseGen}$ in the ideal world and then prove that the resulting transcript is identically distributed to the transcript generated in the real world.

The simulator $\mathcal S_{\NoiseGen}$ interacts externally with $\mathcal F_{\NoiseGen}$ and internally invokes the hybrid-model adversary $\mathcal A$. It works mostly the same way as $\mathcal S_{\InputShare{}}$.
It emulates the honest parties and all calls to $\VSSShare$, $\Multiply$, and $\VSSReconst$.

\begin{enumerate}[leftmargin=1.8em,itemsep=4pt]
    \item The simulator $\mathcal S_{\NoiseGen}$ interacts externally with $\mathcal F_{\NoiseGen}$ and internally invokes the hybrid-model adversary $\mathcal A$.
    \item For every honest $P_i\notin I$, query $q$, and bit position $j$, the simulator chooses an independent uniformly random polynomial $Q_{i,j}^{(0)}\in\mathbb F[X]$ of degree at most $\tau$ satisfying $Q^{(0)}_{i,j}(0)=0.$
    It internally simulates the corresponding $\VSSShare$ invocation and gives $\mathcal A$ the shares $\{Q^{(0)}_{i,j}(\alpha_r)\}_{P_r\in I}$ that the corrupted committee members expect to receive.
    \item It then simulates the multiplication and reconstruction calls and opens the validation value as $0$.
    \item For every corrupted $P_i\in I$, query $q$, and bit position $j$, the simulator extracts the value $\widehat\beta_{i,j}^{(q)}$ fixed by $\VSSShare$.
    If the sharing is missing or rejected, it sets $\widehat\beta_{i,j}^{(q)}=\bot$.
    Otherwise, it simulates the multiplication and reconstruction calls and opens
    \[
    e_{i,j}^{(q)}
    =
    \widehat\beta_{i,j}^{(q)}
    \bigl(1-\widehat\beta_{i,j}^{(q)}\bigr).
    \]
    \item The contribution is accepted if and only if every extracted value are booleans.
    The simulator externally sends the extracted bit vectors to $\mathcal F_{\NoiseGen}$.
    \item The simulator performs the prescribed local aggregation using its simulated sharings.
    For every corrupted $P_r\in I$ and query $q$, let $\rho_{r,q}$ be the resulting aggregate share held by $P_r$.
    The simulator externally sends
    \[
    (\textsc{CorruptedShares},\mathsf{sid},\mathsf{ssid},
    \{\rho_{r,q}\}_{P_r\in I,q\in[N_{iter}]})
    \]
    to $\mathcal F_{\NoiseGen}$.
    \item Finally, the simulator outputs whatever $\mathcal A$ outputs and halts.
\end{enumerate}

We now prove that the transcripts in the real and ideal worlds are identically distributed.
We work first in the hybrid model that includes 
$\VSSShare$,$\Multiply$ and $\VSSReconst$.

The simulator follows $\mathcal S_{\InputShare{}}$ independently for every party $P_i$, query $q$, and bit position $j$.
It simulates every honest bit by VSS-sharing $0$, extracts the values fixed by corrupted dealers, and simulates the corresponding multiplication and reconstruction calls.
It sends the extracted corrupted bit vectors to $\mathcal F_{\NoiseGen}$ and replaces a corrupted contribution by $[0]$ whenever a sharing fails or an extracted value is not a bit.
Finally, it locally aggregates the corrupted parties' simulated shares and sends the resulting output shares $\{\rho_{r,q}\}_{P_r\in I,q\in[N_{iter}]}$ to $\mathcal F_{\NoiseGen}$.

The hybrid execution is compared directly with the ideal execution generated by $\mathcal S_{\NoiseGen}$.
As in the proof of $\InputShare{}$ in Section~\ref{app:inputshare-simulator}, for every honest bit the adversary sees at most $\tau$ evaluations of an independently random degree-$\tau$ Shamir polynomial, whose distribution is independant from the constant term.
Therefore, putting simulator's share of a dummy value does not change the adversary's view.
The corresponding validation values are $0$ in both executions, so the validation transcripts and corrupted aggregate shares are identically distributed as well.

For each query $q$, let $F_q$ be the degree-$\tau$ polynomial obtained by aggregating all accepted bit-sharing polynomials.
Its constant term is the intended noise value $Z_q$.
Conditioned on the adversary's view and the corrupted evaluations
\[
\{F_q(\alpha_r)=\rho_{r,q}\}_{P_r\in I},
\]
the polynomial $F_q$ is uniformly distributed over all degree-$\tau$ polynomials satisfying these constraints, since at least one honest contribution provides independent Shamir randomness.
This is exactly the output distribution prescribed by $\mathcal F_{\NoiseGen}$.

Hence,
\[
\mathsf{HYBRID}_{\NoiseGen,\mathcal A}
\equiv
\mathsf{IDEAL}_{\mathcal F_{\NoiseGen},\mathcal S_{\NoiseGen}}.
\]

The simulator is straight-line, black-box, and non-rewinding.
We then use the same argument as in $\InputShare{}$ security analysis of Section~\ref{app:inputshare-simulator} to prove the security of \NoiseGen{}.
\end{proof}

\subsection{UC Security of \Giskard{}}
\label{app:giskard-proof}

\subsubsection{Ideal Functionality}

The functionality $\mathcal F_{\Giskard}$ is a stateful composition of $\InputShare{}$,$\Reshare{}$,$\NoiseGen{}$,$\VSSReconst{}$ and $\BroadcastDown{}$ and proceeds in $N_{\mathrm{iter}}$ iterations.
It preserves the current search intervals between iterations (as public data) and receives corrupted parties' inputs adaptively after publishing the current pivot.

\subsubsection{Security Analysis}

We now prove Theorem~\ref{thm:giskard-uc} : 

\begin{proof}
Consider the fully hybrid execution in which every invocation of $\InputShare$, $\Reshare$, $\NoiseGen$, $\VSSReconst$, and $\BroadcastDown$ is replaced by its corresponding ideal functionality.
The simulator $\mathcal S_{\Giskard}$ internally invokes $\mathcal A$, forwards its adaptive inputs to $\mathcal F_{\Giskard}$, and simulates the interfaces of the aiding functionalities.

We prove by induction on $t$ that, at the beginning of iteration $t$, the hybrid and ideal executions have identical joint distributions of the adversary's view and the honest parties' state $(\mathsf{Left}_t,\mathsf{Right}_t,\mathbf p_t)$.

For $t=1$, both executions initialize the same search intervals and initial pivot.

Assume the claim holds at the beginning of iteration $t$.
Conditioned on the common transcript, the hybrid and ideal adversaries provide the same adaptive inputs for corrupted leaves.
The ideal $\InputShare$, $\Reshare$, and $\NoiseGen$ calls produce the same distribution of the root committee's shared noisy count in both executions.
The $\VSSReconst$ call therefore produces the same noisy count, after which both executions apply the same deterministic interval update.
Finally, $\BroadcastDown$ delivers the same next pivot to every honest party.
Thus, the induction hypothesis holds for iteration $t+1$.

After $N_{\mathrm{iter}}$ iterations, the hybrid and ideal executions produce the same final midpoint.
Therefore,
\[
\mathsf{HYBRID}_{\Giskard,\mathcal A}
\equiv
\mathsf{IDEAL}_{\mathcal F_{\Giskard},\mathcal S_{\Giskard}}.
\]

The preceding induction establishes perfect security in the fully hybrid model.
Since all aiding protocols perfectly UC-realize their corresponding functionalities, the UC composition theorem~\cite{canettiUC} permits their ideal calls to be replaced by their real implementations.
Consequently, $\Giskard{}$ perfectly UC-realizes $\mathcal F_{\Giskard}$.

Guaranteed output delivery follows from the guaranteed delivery of every aiding functionality and the deterministic nature of the remaining local computations.
\end{proof}

\section{Robustness, Differential Privacy and Complexity Proofs}

\subsection{Proof of Theorem~\ref{thm:giskard-robustness}}
\label{app:robustness-proof}

We first prove the following lemma 
\begin{lemma}[Output bracket]
\label{lem:quantile-bracket}
Set $k_1 = \lfloor n/2\rfloor + 1 - f$ and $k_2 = \lfloor n/2\rfloor + 1$.
For any adaptive Byzantine strategy, the output satisfies $\hat{M} = v + \delta$ with $v \in [h_{(k_1)}, h_{(k_2)}]$ and $|\delta| \leq q/2$.
\end{lemma}

\begin{proof}
Recall $\tau = \lfloor n/2 \rfloor + 1$, so $\tau = k_2$ and $\tau - f = k_1$; since $f < n/2$, both $h_{(k_1)}$ and $h_{(k_2)}$ are well defined.

Let $C_H(p) := |\{j \in \mathcal{H} : x_j \leq p\}|$, as in
\eqref{eq:count-decomposition}.
We claim that $C_H(\mathsf{Left}) < k_2$ and $C_H(\mathsf{Right}) \geq k_1$ throughout.

Initially $C_H(-u) = 0$ (honest inputs exceed $-u$) and $C_H(u) = n - f \geq k_1$.
When $\mathsf{Left} \gets p_t$, the observed $C(p_t) < \tau$ and $C_B(p_t) \geq 0$ give $C_H(p_t) < k_2$; when $\mathsf{Right} \gets p_t$, $C(p_t) \geq \tau$ and $C_B(p_t) \leq f$ give $C_H(p_t) \geq k_1$.

Each step uses only the per-query bound $0 \leq C_B(p_t) \leq f$, enforced by the bit-domain proof, so the invariant holds in all cases.
The invariant means $\mathsf{Left} < h_{(k_2)}$ and
$h_{(k_1)} \leq \mathsf{Right}$.

At termination $\mathsf{Right} - \mathsf{Left} \leq q$, so $\hat{M} = (\mathsf{Left}+\mathsf{Right})/2$ lies within $q/2$ of the nonempty interval $[h_{(k_1)}, h_{(k_2)}]$; taking $v$ as its projection onto this interval gives $|\delta| \leq q/2$.
\end{proof}

We then get to the main robustness proof : 

\begin{proof}
Assume $n > 2f$.
Let $\sigma_{\mathcal{H}}^2 := \tfrac{1}{|\mathcal{H}|}\sum_{j\in \mathcal{H}}(x_j-\bar{x}_{\mathcal{H}})^2$ and set $k_1 = \lfloor n/2\rfloor + 1 - f$, $k_2 = \lfloor n/2\rfloor + 1$.
Since $n > 2f$, we have $1 \leq k_1$, $k_2 \leq n - f$, and $\lceil n/2\rceil - f \geq 1$, so all order statistics below exist and $\kappa_0$ is finite.
By the output bracket of Lemma~\ref{lem:quantile-bracket}, write $\hat{M} = v + \delta$ with $v \in [h_{(k_1)}, h_{(k_2)}]$ and $|\delta| \leq q/2$.
For any $v \in [h_{(k_1)}, h_{(k_2)}]$,
\begin{align*}
|\{j\in \mathcal{H}: x_j \leq v\}| &\;\geq\; k_1, \\
|\{j\in \mathcal{H}: x_j \geq v\}| &\;\geq\; |\mathcal{H}| - k_2 + 1 \;=\; \lceil n/2\rceil - f.
\end{align*}
If $v \geq \bar{x}_{\mathcal{H}}$, each of the $\lceil n/2\rceil - f$ values with $x_j \geq v \geq \bar{x}_{\mathcal{H}}$ contributes $(x_j - \bar{x}_{\mathcal{H}})^2 \geq (v - \bar{x}_{\mathcal{H}})^2$, so
\begin{equation*}
(\lceil n/2\rceil - f)(v - \bar{x}_{\mathcal{H}})^2
\;\leq\; \sum_{j\in \mathcal{H}}(x_j - \bar{x}_{\mathcal{H}})^2
\;=\; |\mathcal{H}|\,\sigma_{\mathcal{H}}^2.
\end{equation*}
If $v \leq \bar{x}_{\mathcal{H}}$, each of the at least $k_1$ values with $x_j \leq v \leq \bar{x}_{\mathcal{H}}$ contributes $(x_j - \bar{x}_{\mathcal{H}})^2 \geq (v - \bar{x}_{\mathcal{H}})^2$, and $k_1 \geq \lceil n/2\rceil - f$ holds for all $n$ since $\lfloor n/2\rfloor + 1 \geq \lceil n/2\rceil$.
In both cases,
\begin{equation}
(v - \bar{x}_{\mathcal{H}})^2 \;\leq\; \kappa_0\,\sigma_{\mathcal{H}}^2,
\quad
\kappa_0 := \frac{n-f}{\lceil n/2\rceil - f}.
\label{eq:chebyshev-bound}
\end{equation}
By the triangle inequality and \eqref{eq:chebyshev-bound},
\begin{equation}
|\hat{M} - \bar{x}_{\mathcal{H}}|
\;\leq\; |v - \bar{x}_{\mathcal{H}}| + |\delta|
\;\leq\; \sqrt{\kappa_0}\,\sigma_{\mathcal{H}} + \tfrac{q}{2},
\label{eq:additive-bound}
\end{equation}
\end{proof}

\subsection{Proof of Theorem~\ref{thm:dp-property}}
\label{app:proof-dp-property}

\begin{proof}
    This proof will proceed in four steps : we first derive the ideal zCDP guarantee, we then account for the discretization of the Gaussian noise, we then convert it into $(\epsilon,\delta)$-DP and account for the truncation error.

    \begin{enumerate}
        \item The count queries are of senstivity-$1$ (as one party may only shift the count by $1$)n then for $T$ training rounds and $d$ update coordinate, so adding noise from $\mathcal{N}\left(0,\frac{TdN_{\mathrm{iter}}}{2\rho}\right)$ gives directly $\rho /TdN_{\mathrm{iter}}$-zCDP for each iteration. Finally with the additive composition theorem we get the $\rho$-zCDP guarantee.
        \item We use Proposition 13 in \cite{kairouz2021distributed} as we work with coordinate-wise aggregation (so each coordinate are considered independently). Then apply their formula for sensitivity 1, we have $A := \frac{1}{2h\sigma_c^2}$, then by noting $B - A = \frac{\vartheta_{\mathrm{DG}}}{\sqrt h\sigma_c}+\frac{\vartheta_{\mathrm{DG}}^2}{2}$
        by direct expansion, we get $\min\{A+\vartheta_{\mathrm{DG}}, B\} = A + \min\{\vartheta_{\mathrm{DG}}, B-A\}$
        which follows from $\min(a+c,a+d)=a+\min(c,d)$.
                \item
        Combining the previous two steps and composing over all $TdN_{\mathrm{iter}}$ count queries gives
        \begin{equation}
        \bar{\rho}
        :=
        \rho+\xi_{\mathrm{DG}}.
        \end{equation}
        Hence, the ideal untruncated mechanism satisfies $\bar{\rho}$-zCDP.
        By the standard zCDP to $(\epsilon,\delta)$-DP conversion, for any $\delta_0\in(0,1)$, it satisfies $(\varepsilon_0,\delta_0)$-DP with
        \begin{equation}
        \varepsilon_0
        =
        \bar{\rho}
        +
        2\sqrt{\bar{\rho}\ln(1/\delta_0)}.
        \end{equation}

        \item
        Each honest contribution produced by rejection sampling follows the discrete conditional Gaussian situated on $\{-B,\ldots,B-1\}$.
        Since the ideal discrete Gaussian assigns probability at most $q_B$ outside this domain, the total variation distance between one truncated sample and one ideal sample is at most $q_B$. Over the $TdN_{\mathrm{iter}}h$ honest contributions, the total variation distance is therefore at most
        \begin{equation}
        \beta_{\mathrm{tr}}
        =
        \min\left\{
        1,\,
        TdN_{\mathrm{iter}}h q_B
        \right\}.
        \end{equation}

        For the rank error, the aggregate honest noise has variance $\sigma^2$ and satisfies
        \begin{equation}
        \Pr\!\left[|Z|>\eta_{\mathrm{DP}}\right]
        \le
        2\exp\!\left(
        -\frac{\eta_{\mathrm{DP}}^2}{2\sigma^2}
        \right)
        =
        \frac{\beta}{TdN_{\mathrm{iter}}}.
        \end{equation}
        A union bound over the $TdN_{\mathrm{iter}}$ noisy count queries therefore gives an additional rank error of at most $\eta_{\mathrm{DP}}$ simultaneously over all coordinates with probability at least $1-\beta$.
        Combining this event with the probability of no truncation gives the claimed probability at least $1-\beta-\beta_{\mathrm{tr}}$.
        \item
        Furthermore, We must consider the breakdown point of our aggregation function at $(f < n/2)$ and is equivalent to the number of corrupted inputs that this function can tolerate. By relating this to lemma~\ref{lem:quantile-bracket}, we see that this affects the entire range of values that the median can take. Since the DP mechanism introduces a rank error, we simply need to treat it as a Byzantine fault to be tolerated, which reduces the effective tolerance for inputs to 
        \begin{equation}
            f_{DP} = f - \eta_{DP}
        \end{equation}
        Furthermore, considering that in the root committee there may be up to $\tau$ corrupted parties, they may end up shifting the rank by at most $\tau B$ rank, thus getting the final effective tolerance of this robust aggregation function
    \end{enumerate}
\end{proof}

\subsection{Proof of Theorem~\ref{thm:treeConstruction}}
\label{app:treeconstruction-proof}

\begin{proof}
    Corruption is static and fixed before $\mathsf{seed}$ is sampled.
    Since $\mathsf{seed}$ is unpredictable, each committee is a uniformly random $m$-subset of $[n]$, so its corruption count follows $\mathrm{Hyp}(n, f, m)$.
    By Hoeffding's reduction~\cite{Hoeffding63}, upper-tail bounds for this hypergeometric are upper bounded by those of $X \sim \mathrm{Bin}(m, f/n)$.
    We bound $\Pr[X \geq m/4]$ via a multiplicative Chernoff bound:
    \begin{align*}
    \Pr[X \ge m/4] \;\le\; \exp\!\left( -\left(1 - \frac{3n}{4(n-f)}\right)^{\!2} \cdot \frac{m(n-f)}{2n} \right).
    \end{align*}
    Substituting $f/n = 1/4 - \epsilon$,
    \begin{align*}
    \Pr[X \ge m/4]
    &\leq \exp\!\left(-\frac{16\epsilon^2}{(3+4\epsilon)^2} \cdot \frac{m(3/4 + \epsilon)}{2}\right) \\
    &= \exp\!\left(-\frac{2\epsilon^2 m}{3 + 4\epsilon}\right).
    \end{align*}
    The total number of committees is $O(n)$, so a union bound gives:
    \begin{align*}
    \Pr[\text{failure}] \;\leq\; n \cdot \exp\!\left(-\frac{2\epsilon^2 m}{3 + 4\epsilon}\right).
    \end{align*}
    For $\Pr[\text{failure}] \leq n^{-c}$ with any desired constant $c > 0$, taking logarithms:
    \begin{align*}
    \ln n - \frac{2\epsilon^2 m}{3 + 4\epsilon} &\leq -c \ln n \\
    m &\geq \frac{(1+c)(3+4\epsilon)}{2\epsilon^2}\ln n.
    \end{align*}
    Since $\epsilon \in (0, 1/4)$ is fixed, this gives $m = O(\epsilon^{-2} \ln n)$.
\end{proof}

\subsection{Proof of Theorem~\ref{thm:comm-bgw}}
\label{app:comm-proof}
Before starting the proof, lets show the communicaton costs of each sub-protocol taken in isolation~\cite{bgwSubshare}:
\begin{itemize}
\item \textbf{\textsc{VSS-Share}} (single instance, $m$ parties): $O(m^2)$ per party and $O(m^2)$ total network communication.
\emph{(Non-optimistic cost : $O(m^3)$ per party and $O(m^4)$ total.)}

\item \textbf{\textsc{Multiply}} ( $\tau<m/4$): $O(m^2)$ per party via $\VSSSubShare$ applied directly to the local products $a_i b_i$, with total network cost $O(m^3)$.
\emph{(Non-optimistic cost: $O(m^4)$ per party and $O(m^5)$ total.)}

\item \textbf{\textsc{InputShare}}: a committee party handles $n/k^L$ leaves, each requiring one \textsc{VSS-Share} and one \textsc{Multiply} for bit verification, giving $O\!\left((n/k^L)\cdot m^2\right)$ per party.
\emph{(Non-optimistic cost: $O\!\left((n/k^L)\cdot m^4\right)$ per party.)}
Local aggregation of validated shares is communication-free by additive homomorphism.

\item \textbf{\textsc{Reshare}}: since resharing uses $\VSSSubShare$, it costs $O(m^2)$ per party on both sides and $O(m^3)$ in total per child-to-parent transfer, while a parent receiving from $k$ children pays $O(km^2)$ per party.
\emph{(Non-optimistic cost: $O(m^4)$ per party on both sides and $O(m^5)$ total per transfer, while a parent receiving from $k$ children pays $O(km^4)$ per party.)}

\item \textbf{\textsc{VSS-Reconst}} (committee output): $O(m)$ per party.
\emph{(Non-optimistic cost: unchanged at $O(m)$ per party.)}

\end{itemize}
\begin{proof}
    Let $B=k^L$ denote the number of base committees and let $b=\left\lceil \frac{n}{B}\right\rceil$ be the maximum number of leaves served by one base committee.
    We have the constraint that $Bm\le n$ as the committee are disjoint within a tree level. 
    In the best configuration we then have
    \[
    B=\Theta(n/m),\qquad b=\Theta(m),
    \qquad
    L=\Theta\!\left(\log_k(n/m)\right).
    \]
    For constant $k$ and $m=\Theta(\log n)$, this implies $L=\Theta(\log n)$.
    
    For part (i) we assume that in the worst-case a party may participate as one leaf role, one base role, at most $L-1$ intermediate committee roles, and one root committee role. Plugging optimistic sub-protocol cost, the worst case communication per coordinate and per iteration is bounded by 
    Hence its communication
    per coordinate and per iteration is bounded by
    \[
    O\!\left(
    m^2
    +bm^2
    +Lkm^2
    +qm^3
    +Lkm
    \right),
    \]
    where the last term accounts for broadcasting down the public output thought the tree. With $b =O(km)$ we get 
    \[
    O\!\left(
    km^3+Lkm^2+qm^3
    \right).
    \]
    Then, with $k,q=O(1)$ and $m,L=\Theta(\log n)$ we have 
    \[
    \Delta_{\max}^{\mathsf{eff}}
     =
    O\!\left(dN_{\mathrm{iter}}\log^3 n\right).
    \]
    Plugging non optimistic sub-protocols cost yield 
    \[
    \Delta_{\max}^{\mathsf{eff}}
     =
    O\!\left(dN_{\mathrm{iter}}\log^5 n\right).
    \]
    
    For part (ii) we consider total communication over the entire network. The base level contains $Bm = \Theta(n/m)m=\Theta(n)$ distinct committee member, each handling the output from $\Theta(m)$ leaves. With the cost of \InputShare it means that at the base level the total cost  per iteration is $dnm^3$.
    Then for the cost of \Reshare{}, we count the number of committee edges to get the number of \Reshare{} instance. We get that this is  
    \[
    \sum_{j=0}^{L-1} k^j
     = O(k^L)
     = O(n/m).
    \]
    As such, plugging the cost of \Reshare{} we get 
    \[
    O\!\left(d\frac{n}{m}m^3\right)
     = O(dnm^2).
    \].
    For the broadcast down, we still use the edge count, getting
    \[
    O\!\left(d\frac{n}{m}m^2\right)
     = O(dnm).
    \].
    Finally, \NoiseGen{} costs $O(qm^3)$ per root member and therefore $O(dqm^4)$ globally.
    
    We then see that \InputShare{} cost dominates asymptotically , yielding 
    \[
    O\!\left(
    dN_{\mathrm{iter}}n\log^3 n
    \right).
    \]
    Finally for non optimistic execution,
    \[
    \Delta_{\max}^{\mathsf{eff}}
     =
    O\!\left(dN_{\mathrm{iter}}\log^5 n\right).
    \]
\end{proof}

\end{document}